\documentclass[12pt]{article}

\usepackage{amsmath,amsthm,amsfonts,amssymb,mathtools}
\usepackage{float}
\usepackage[svgnames,dvipsnames]{xcolor}
\usepackage{enumitem}
\usepackage{graphicx}
\usepackage{epstopdf}
\usepackage{tikz-cd}
\usepackage{tcolorbox}

\usepackage[all]{xy}
\usepackage{pb-diagram,pb-xy}

\definecolor{dullmagenta}{RGB}{102,0,102}
\def\col{dullmagenta} %color para links %Blue,Maroon

\usepackage[
    colorlinks=true,       % false: boxed links; true: colored links
    linkcolor=\col,          % color of internal links
    citecolor=\col,      % color of links to bibliography
    filecolor=\col,       % color of file links
    urlcolor=\col,            % color of external links
    ]{hyperref}

\makeatletter %%removes too much space before \paragraph
\renewcommand\paragraph{\@startsection{paragraph}{4}{\z@}%
                                      {\parskip}%{3.25ex \@plus1ex \@minus.2ex}%
                                      {-1em}%
                                      {\normalfont\normalsize\bfseries}}
\makeatother

\parskip=\medskipamount
\font\frak=eufm10 %scaled\magstep1
\def\goth #1{\hbox{{\frak #1}}}
\def\lag{{\goth{g}}}

\newcommand{\g}{\goth{g}}

\newcommand{\longto}{\longrightarrow}
\newcommand{\cL}{{\pounds}}

\def\fpd#1#2{{\displaystyle\frac{\partial #1}{\partial #2}}}

\newcommand{\R}{\mathbb{R}}

\newcommand{\Lin}{{\rm Lin}}

\makeatletter %% To avoid paragraph break after theorem unless new paragraph starts!
\def\@endtheorem{\endtrivlist}% NEW
\makeatother

\makeatletter
\def\th@plain{%
  \thm@notefont{}% same as heading font
  \itshape % body font
}
\def\th@definition{%
  \thm@notefont{}% same as heading font
  \normalfont % body font
}
\makeatother

\theoremstyle{plain}
\newtheorem{theorem}{Theorem}
\newtheorem{lemma}{Lemma}
\newtheorem{proposition}{Proposition}
\newtheorem{corollary}{Corollary}

\theoremstyle{definition}
\newtheorem{remark}{Remark}

\newtheorem{definition}{Definition}

\begin{document}
\date{}
\title{Routh reduction for first-order field theories}

\author{S.\ Capriotti\textsuperscript{a,b} and  E.\ Garc\'{\i}a-Tora\~{n}o Andr\'{e}s\textsuperscript{a}\\[2mm]
{\small \textsuperscript{a} Departamento de Matem\'atica,
Universidad Nacional del Sur}  \\
{\small  Av.\ Alem 1253, 8000 Bah\'ia Blanca, Argentina}\\[2mm]
{\small \textsuperscript{b} Instituto de Matem\'atica de Bah\'ia Blanca (INMABB), CONICET}}

\maketitle

\begin{abstract}

We present a reduction theory for first order Lagrangian field theories which takes into account the conservation of momenta. The relation between the solutions of the original problem with a prescribed value of the momentum and the solutions of the reduced problem is established. An illustrative example is discussed in detail. 

\vspace{3mm}


\vspace{3mm}

\end{abstract}

\section{Introduction}\label{sec:Intro}

The geometric reduction of an invariant Lagrangian system can be performed in two different ways depending on whether the conservation of momenta is taken into account in the reduction procedure or not. Accordingly, two fundamental reduction theories have emerged in the literature: the so-called \emph{Routh reduction} and the \emph{Lagrange-Poincar\'e reduction}. In a nutshell, the distinction is as follows: in the Lagrange-Poincar\'e reduction one quotients the tangent bundle of the phase space directly by the Lie group of symmetries, while in the case of Routh reduction one first restricts the attention to the level set of the momentum and only then quotients by a suitable subgroup of the group of symmetries (in fact, Routh reduction is the natural Lagrangian analog of symplectic reduction~\cite{quasi}). As the terminology suggest, both techniques have a rich and long history; the reader can take a look at~\cite{RouthMarsden} for an overview.

In the realm of Lagrangian field theory, much attention has been paid to the Lagrange-Poincar\'e case (see e.g.~\cite{Castrillon1,Castrillon2} and references therein; see also~\cite{GbRat,LagPoincare_JGP}) but, to the best of our knowledge, the case of Routh reduction remains unexplored. The purpose of this paper if to fill this gap and discuss a variational framework to carry out Routh reduction for first-order Lagrangian field theories. This is achieved by examining the contact structure of the variational problem with symmetry, which allows us to relate the critical sections of the original Lagrangian field theory with a prescribed value of the momentum with the critical sections of a reduced Lagrangian field theory with forces. The role of the reduced Lagrangian is played by a suitably defined \emph{Routhian}, extending the well-known construction for the mechanical case (see e.g.~\cite{Marsden_Lectures}).

The paper is organized as follows.  In Section~\ref{sec:LFT} we discuss the setting of Lagrangian field theory (LFT) and identify the critical sections of a Lagrangian density as integral sections of a suitable affine subbundle of the contact subbundle. We describe how the symmetry (and the choice of a principal connection) induces a splitting of the contact structure. A notion of momentum map adapted to this setting is introduced following~\cite{GotayIM_Momentum}, and its conservation along solutions is shown. Section~\ref{sec:RouthLFT} contains the main results on Routh reduction for field theories. First, we identify a natural candidate for a Routhian in the field-theoretical setting and then, after some preparatory results, we prove that its reduction plays the role of the Lagrangian for a reduced Lagrangian field theory with forces. It is shown that extremals of the original (i.e. unreduced) LFT with a prescribed value of the momentum project onto solutions of the reduced LFT. Many of the results in this section are obtained adapting the techniques from~\cite{Capriotti_Routh} to an arbitrary configuration bundle. Section~\ref{sec:reconstruction} addresses the problem of reconstruction. We recover the integrability condition for reconstruction that has appeared in the context of Lagrange-Poincar\'e reduction~\cite{Castrillon2,LagPoincare_JGP} and discuss its geometric meaning in terms of liftings of sections. Finally, Section~\ref{sec:Example} contains one easy example that illustrates the applicability of the proposed scheme.

To conclude, we would like to point that it should be possible to consider a different approach to Routh reduction in field theory, at least under some regularity conditions, in those formalisms for which a multisymplectic-like reduction theorem is available. For instance the case of polysymplectic manifolds, arguably one of the easiest approach to field theories, has its own reduction theorem~\cite{Poly} and thus looks like a natural first choice. This will be discussed elsewhere.

\paragraph{Notations.} If $Q$ is a manifold, $\Lambda^p (Q)=\wedge^p(T^*Q)$ denotes the $p$-th exterior power of the cotangent bundle of $Q$. The space of differential $p$-forms, sections of $\Lambda^p (Q)\to Q$, will be denoted by $\Omega^p(Q)$. We also write $\Lambda^\bullet(Q)=\bigoplus_{j=1}^{\dim Q}\Lambda^j(Q)$. If $f\colon P\to Q$ is a smooth map and $\alpha_x$ is a $p$-covector on $Q$, we will sometimes use the notation $\alpha_{f(x)}\circ T_xf$ to denote its pullback $f^*\alpha_x$. If $P_1\to Q$ and $P_2\to Q$ are fiber bundles over the same base $Q$ we will write $P_1\times_Q P_2$ for their fibred product, or simply $P_1\times P_2$ if there is no risk of confusion. Finally, Einstein summation convention will be used everywhere.

\paragraph{Acknowledgements}
S.\ Capriotti thanks the CONICET for financial support.

\section{Lagrangian field theory}\label{sec:LFT}

We will denote the configuration (fibre) bundle by $\pi\colon E\to M$, with $\dim M=m$ and $\dim E=m+n$, and we assume that $M$ is oriented with volume form $\eta$. We consider the first jet bundle $J^1\pi$ and adopt the usual notations for the source and target projections:
\begin{equation*}
  \begin{tikzcd}[column sep=.8cm, row
    sep=1.2cm]
    J^1\pi \arrow[rr,"\pi_{10}"]\arrow[dr,swap,"\pi_1"] & &[1ex] E\arrow[dl,"\pi"]\\
    & M &
  \end{tikzcd}
\end{equation*}
A section of $\pi$ will generally be denoted by $s:M\to E$, $j_x^1 s$ denotes the first jet of $s$ at $x\in M$ and $j^1s:M\to J^1\pi$ denotes the prolongation of $s$. We will use adapted coordinates $(x^i,u^a,u^a_i)$ on $J^1\pi$ such that locally $\eta=dx^1\wedge\dots\wedge dx^m$. For $p<l$, the set of $p$-horizontal $l$-forms on $E$, denoted $\Lambda^l_kE$, is defined as
\[
\Lambda^l_p E\big|_e=\{\alpha\in \Lambda^l E: v_1\lrcorner\dots v_p \lrcorner \alpha=0, \quad \text{for all } v_1,\dots,v_p\in V_e(\pi)\}.
\]
Likewise, the set
\[
\Lambda^l_p J^1\pi \big|_{j^1_xs}=\{\alpha\in \Lambda^l J^1\pi\colon  v_1\lrcorner\dots v_p \lrcorner \alpha=0, \quad \text{for all } v_1,\dots,v_p\in V_{j^1_xs}(\pi_1)\}
\]
denotes the $p$-horizontal $l$-forms on $J^1\pi$. For the necessary background on the geometry of first-order Lagrangian field theory, we refer the reader to~\cite{GeometryLagrangianFieldTheories}. A comprehensive treatment of jet bundles can be found in~\cite{Saunders_book}.

We will need the following definition of Lagrangian field theory which is more general than the standard one:

\begin{definition}\label{def:LFT} A \emph{Lagrangian field theory (LFT)} is a triple $(\pi\colon E\rightarrow M,L\eta,\mathcal{F})$, where $L$ is a smooth function on $J^1\pi$, $\eta$ is the pullback to $J^1\pi$ of a volume form on $M$ and $\mathcal{F}\in\Omega^{m+1}_3(J^1\pi)$ is a $\pi_{10}$-basic $(m+1)$-form on $J^1\pi$.
\end{definition}
The $\pi_1$-semibasic form $\mathcal{L}=L\eta$ is the \emph{Lagrangian density}, and $\mathcal{F}$ is the \emph{force}. We will make a slight abuse of notation and write $\eta$ for both the volume form on $M$ and its pullback to $J^1\pi$ or any space that fibers over $M$. The case $\mathcal{F}=0$ corresponds to the usual definition of LFT and will simply be represented by the pair $(\pi\colon E\rightarrow M,L\eta)$. 

We say that a (possibly local) section $s\colon U\subset M\to E$ is \emph{critical} for $(\pi\colon E\rightarrow M,L\eta)$ if
  \[
  \delta\int_U\left(j^1s\right)^*L\eta+\int_U\left(j^1s\right)^*\langle \mathcal{F},\delta s^C\rangle=0
  \]
holds for every variation $\delta s$ that vanishes on $\partial U$, where $\delta s^C$ is the section of the pullback bundle $(j^1s)^*\left(V\pi_1\right)$ constructed by the complete lift of an extension of the section $\delta s:M\rightarrow V\pi$ to a vertical vector field on $E$. In coordinates, writing $\mathcal{F}=\frac{1}{2}F^j_{ab} du^a \wedge du^b \wedge\eta_j + F_adu^a\wedge\eta$ (here $\eta$ denotes the pullback to $J^1\pi$ of the volume form in $M$), a section $x^i\mapsto (x^i,u^a(x))$ is critical iff it satisfies:
\[
\fpd{}{x^k}\left(\fpd{L}{u^a_k}\right)-\fpd{L}{u^a}=F^j_{ab}\fpd{u^b}{x^j}+F_a.
\]
Note that the case $\mathcal{F}=0$ leads to the well-known Euler-Lagrange equations for first-order fields.

\subsection{Contact structure and critical sections}

We will now describe an alternative characterization of the solutions of a Lagrangian field theory $(E,L\eta,\mathcal{F})$. The approach is based on the notion of classical Lepage-equivalent variational problems~\cite{Gotay91} (see also~\cite{Krupka_book} and references therein) and the Griffiths formalism~\cite{Griffiths_book} (see also~\cite{Hsu_variations}). 

The canonical form on $J^1\pi$ is the $V\pi$-valued 1-form $\theta=(du^a-u^a_i dx^i)\otimes \partial_{u^a}$, which can be intrinsically expressed as (see e.g.~\cite{GeometryLagrangianFieldTheories}):
\begin{equation}\label{eq:contactstructuredefinition}
\left.\theta\right|_{j_x^1s}=T_{j_x^1s}\pi_{10}-T_xs\circ T_{j_x^1s}\pi_1.
\end{equation}
We consider the contact bundle $I_{\rm con}$ which is the subbundle of $\Lambda^\bullet (J^1\pi)$ generated by the forms $\theta^a=du^a-u^a_i dx^i$. With this it is meant that an element in the contact bundle $\rho\in I_{\rm con}$ is of the form
\[
\rho= \theta^a\wedge\beta_a, \qquad \beta_a\in\Lambda^\bullet(J^1\pi).
\]
Since $\theta$ is a $V\pi$-valued 1-form, composing $\theta$ with a section $\alpha$ of  $V^*\pi$ results in a 1-form on $J^1\pi$ which is a combination of the forms $\theta^a$, and therefore we can think of $I_{\rm con}$ as the subbundle generated (in the sense above) by the forms $\alpha\circ\theta$ with $\alpha$ a section of $V^*\pi$. In our case, it will be convenient to think of $I_{\rm con}$ as generated by the forms $\alpha\circ\theta$ where $\alpha$ is a 1-form on $E$ (again, this holds because $\theta$ is $V\pi$-valued).

We consider the \emph{contact subbundle} $I_{{\rm con},2}^m=I_{\rm con}\cap\Lambda^m_2 J^1\pi$ spanned by $m$-forms which are 2-horizontal and which, in view of the observations above, admits the following description:
\begin{equation}\label{eq:ContactFields}
  \left.I_{{\rm con},2}^m\right|_{j_x^1s}=\mathcal{L} \left\{\alpha\circ(T_{j_x^1s}\pi_{10}-T_xs\circ T_{j_x^1s}\pi_1)\wedge\beta:\alpha\in T^*_{s\left(x\right)}E, \beta\in\left(\Lambda^{m-1}_1J^1\pi\right)_{j_x^1s}\right\}.
\end{equation}
The notation $\mathcal{L}\{\cdot \}$ denotes the linear span. In other words, an element $\rho$ in the contact subbundle $I_{{\rm con},2}^m$ is of the form
\begin{equation*}
 \rho=(\alpha_1\circ\theta)\wedge\beta_1+\dots+(\alpha_k\circ\theta)\wedge\beta_k,
\end{equation*}
sor some $k\in\mathbb{N}$ and with $\alpha_i,\beta_i$ as in~\eqref{eq:ContactFields}. We will call an element of $I_{{\rm con},2}^m$ with a single summand (i.e., $k=1$) a \emph{simple element}. Most of the proofs involving $I_{{\rm con},2}^m$ will be done for simple elements, since the case of arbitrary elements is similar. 

Finally, we consider the subbundle $W_{L\eta}$ of $\Lambda^m_2J^1\pi$ given by the affine translation of the contact subbundle by the Lagrangian density:
\begin{equation}\label{eq:affinesubbundle}
 W_{L\eta}=L\eta+I^m_{{\rm con},2}\subset\Lambda^m_2J^1\pi.
\end{equation}
We denote by $\pi_{L\eta}:W_{L\eta}\rightarrow J^1\pi$ its canonical projection. Coordinates on $W_{L\eta}$ are given as follows. The bundle $I_{{\rm con},2}^m$ is spanned by the forms $\gamma^a_i=\theta^a\wedge \eta_i$, with $\eta_i=\partial_{x^i}\lrcorner \eta$, and thus an element $\alpha_{j_x^1s}\in (W_{L\eta})_{j_x^1s}$ is expressed as $\alpha_{j_x^1s}=(L\eta)_{j_x^1s}+p_a^i(\gamma^a_i)_{j_x^1s}$ for some multipliers $p^i_a$. This defines coordinates on the fibers of $\pi_{L\eta}$, and therefore $(x^i,u^a,u^a_i,p^i_a)$ are coordinates on $W_{L\eta}$ which are adapted to the fibrations:
\begin{equation*}
\begin{tikzcd}[column sep= 1.2cm, row sep = .2cm]
W_{L\eta} \arrow[r,"\pi_{L\eta}"] & J^1\pi\arrow[r,"\pi_{10}"] & E\arrow[r,"\pi"] &M,\\
(x^i,u^a,u^a_i,p^i_a)\arrow[r,mapsto] & (x^i,u^a,u^a_i)\arrow[r,mapsto] & (x^i,u^a)\arrow[r,mapsto] & x^i.
\end{tikzcd}
\end{equation*}

The bundle $\pi_{L\eta}\colon W_{L\eta}\to J^1\pi$ comes equipped with a corresponding Cartan $m$-form $\lambda_{L\eta}$. It is defined as follows: for all $v_1,\dots,v_m\in T_\alpha W_{L\eta}$, 
\begin{equation*}
\left.\lambda_{L\eta}\right|_\alpha (v_1,\dots,v_m)=\pi_{L\eta}^*\alpha(v_1,\dots,v_m)=\alpha\big(T\pi_{L\eta}(v_1),\dots, T\pi_{L\eta}(v_m) \big).
\end{equation*}
In coordinates, it reads:
\[
\lambda_{L\eta}=L\eta+p_a^i du^a\wedge \eta_i- p_a^i u^a_i \eta.
\]
We can now prove a useful characterization of the critical sections of a LFT:

\begin{proposition}\label{prop:FieldTheoryEqsWL}
A section $s\colon U\subset M\rightarrow E$ is critical for $(\pi\colon E\rightarrow M,L\eta)$ if and only if there exists a section $\Gamma\colon U\subset M\rightarrow W_{L\eta}$ such that
\begin{enumerate}[label={\arabic*)},nolistsep]
\item $\Gamma$ covers $s$, i.e. $\pi_{10}\circ\pi_{L\eta}\circ\Gamma=s$, and
\item $\Gamma^*\left(X\lrcorner d\lambda_{L\eta}\right)=0$, for all $X\in\mathfrak{X}^{V(\pi_1\circ\pi_{L\eta})}(W_{L\eta})$.
\end{enumerate}
\end{proposition}
Here $\mathfrak{X}^{V(\pi_1\circ\pi_{L\eta})}(W_{L\eta})$ denotes the vector fields which are vertical w.r.t. the projection $W_{L\eta}\to M$. $\Gamma$ is called a \emph{solution} of $(\pi\colon E\rightarrow M,L\eta)$ or of $(W_{L\eta},\lambda_{L\eta})$.
\begin{proof}
The situation is summarized in the following diagram:
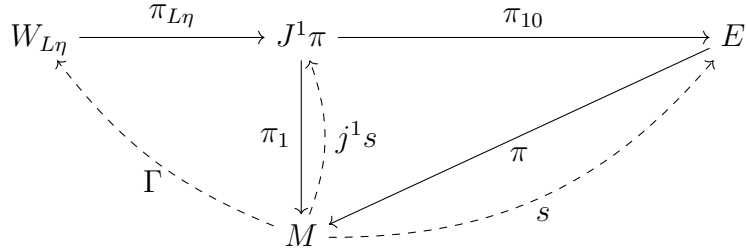
\[
    \begin{tikzpicture}
      \matrix (m) [matrix of math nodes, row sep=5em, column sep=6em,
      text height=1.5ex, text depth=0.25ex]
      { W_{L\eta} & J^1\pi & & E   \\
        & M & &\\ }; \path[->] (m-1-2) edge node[left] {$
        \pi_1 $} (m-2-2) (m-1-1) edge node[above] {$ \pi_{L\eta} $}
      (m-1-2) (m-1-2) edge node[above] {$ \pi_{10} $} (m-1-4) (m-1-4)
      edge node[below] {$ \pi $} (m-2-2); \path[->,dashed]
      (m-2-2) edge [bend left=15] node[below] {$ \Gamma $} (m-1-1)
      edge [bend right=20] node[right] {$ j^1s $} (m-1-2) edge [bend
      right=25] node[below] {$ s $} (m-1-4);
    \end{tikzpicture}
\]
Considering the vector fields $\{\partial_{u^a},\partial_{u^a_i},\partial_{p^i_a}\}$ the condition $\Gamma^*\left(X\lrcorner d\lambda_{L\eta}\right)=0$ translates into:
\begin{align*}
    0&=\Gamma^*\left(\frac{\partial}{\partial u^a}\lrcorner d \lambda_{L\eta}\right)=\Gamma^*\left(\frac{\partial L}{\partial u^a}\eta- dp^k_a\wedge\eta_k\right),\\
    0&=\Gamma^*\left(\frac{\partial}{\partial u^a_k}\lrcorner d \lambda_{L\eta}\right)=\Gamma^*\left(\left(\frac{\partial L}{\partial u^a_k}-p_a^k\right)\eta\right),\\
    0&=\Gamma^*\left(\frac{\partial}{\partial p^k_a}\lrcorner d \lambda_{L\eta}\right)=\Gamma^*\big(\big( d  u^a-u^a_l d  x^l\big)\wedge\eta_k\big).
\end{align*}
Hence $\Gamma=(x^i,u^a(x),u^a_i(x),p^i_a(x))$ must satisfy
\[
\frac{\partial L}{\partial u^a}-\frac{\partial p^k_a}{\partial x^k}=0, \qquad \frac{\partial L}{\partial u^a_k}-p_a^k=0, \qquad \frac{\partial u^a}{\partial x^k}=u^a_k,
\]
which are the Euler-Lagrange equations, written in implicit form.
\end{proof}
In the presence of a force $\mathcal{F}$, a similar proof shows that $s\colon U\subset M\rightarrow E$ is a critical section for $(\pi\colon E\rightarrow M,L\eta,\mathcal{F})$ if and only if there exists a section $\Gamma\colon U\subset M\rightarrow W_{L\eta}$ covering $\gamma$ and such that
\[
\Gamma^*\big(X\lrcorner \big(d\lambda_{L\eta}+\widetilde{\mathcal{F}}\big)\big)=0,
\]
for all $X\in\mathfrak{X}^{V(\pi_1\circ\pi_{L\eta})}(W_{L\eta})$, where $\widetilde{\mathcal{F}}\in\Omega^{m+1}(W_{L\eta})$ is the pullback of $\mathcal{F}$. 

\begin{remark} The fact that the Euler-Lagrange equations obtained from the relation $\Gamma^*\left(X\lrcorner d\lambda_{L\eta}\right)=0$ are implicit has an important consequence:  the momentum constraint will be kept implicit through the reduction procedure, and this helps us to overcome the usual issues related to the group regularity (regularity w.r.t.\ the group variables) of the Lagrangian. In this respect, our approach is similar to~\cite{2016_RouthDirac}.
\end{remark}

In what follows we will consider only global solutions of the LFT, but all the results apply as well to local solutions.

\subsection{Symmetry and momentum}

We now discuss the presence of natural symmetries and their momentum maps for a LFT $(\pi\colon E\rightarrow M,L\eta)$. For concreteness, we will work with left actions.

We start with an action $\phi\colon G\times E\to E$ of a Lie group $G$ on $E$ which is vertical, i.e. $\pi(g e)=\pi(e)$ for each $g\in G$ and $e\in E$, where $ge= \phi_g(e)=\phi(g,e)$. We assume that the action is free and proper, and thus $p_G^E\colon E\rightarrow E/G$ is a principal fiber bundle. We will denote by $\overline{\pi}\colon E/G\rightarrow M$ the quotient bundle:
\begin{equation*}
  \begin{tikzcd}[column sep=.8cm, row
    sep=1.2cm]
   E\arrow[rr,"p_G^E"]\arrow[dr,swap,"\pi"] & &[-2ex] E/G\arrow[dl,"\overline{\pi}"]\\
    & M &
  \end{tikzcd}
\end{equation*}
The infinitesimal generator of an element $\xi\in\lag$ (where $\lag$ is the Lie algebra of $G$) will be denoted by $\xi_E$. More in general, if $Q$ is a manifold with a $G$-action we use the notation $\xi_Q$ for the infinitesimal generators.

There are natural (left) $G$-actions on $J^1\pi$ (by prolongation, i.e. $j^1\phi_g$), on $T(J^1\pi)$ (via the tangent lift) and on $\Lambda^p(J^1\pi)$ for any $p$ (via the cotangent lift). We will often use the abbreviated notation for all of them: for instance, if $\alpha_{j^1_xs}\in \Lambda^m(J^1\pi)$ we write $g\alpha_{j^1_xs}=T^*_{j^1\phi_g(j^1_xs)}j^1\phi_{g^{-1}}\alpha_{j^1_xs}$, and so on. We assume that the action leaves the Lagrangian density invariant. More precisely, we require $(j^1\phi_g)^*L\eta=L\eta$.

In this situation, it can be shown that the action preserves the contact subbundle $I_{{\rm con},2}^m$ and therefore, in view of the invariance of the Lagrangian density, it also preserves the subbundle $W_{L\eta}$. Moreover, the Cartan form $\lambda_{L\eta}$ is invariant w.r.t.\ this action: this can be checked using the argument in~\cite{GotayIM_Momentum}, Section 4.B. In this setting, the notion of momentum map for the action on $W_{L\eta}$ is introduced following~\cite{GotayIM_Momentum}, and it is a particular case of the more general multisymplectic approach~\cite{CaCraIb_1991}.  

\begin{definition}
A \emph{momentum map} for the action of $G$ on $W_{L\eta}$ is a map 
\[
 J\colon W_{L\eta}\to \Lambda^{m-1} W_{L\eta} \otimes \lag^*
\]
over the identity in $W_{L\eta}$ such that 
\[
\xi_{W_{L\eta}}\lrcorner d\lambda_{L\eta}=-dJ_\xi,
\]
where $J_\xi$ is the $(m-1)$-form on $W_{L\eta}$ whose value at $\alpha\in W_{L\eta}$ is $J_\xi(\alpha)=\langle J(\alpha),\xi \rangle$.
\end{definition}

Accordingly, we think of a ``momentum'' $\widehat{\mu}$ as an element $\widehat{\mu}\in \Omega^{m-1}(W_{L\eta},\lag^*)$, i.e. as a $\lag^*$-valued $(m-1)$-form on $W_{L\eta}$; a conserved value $\widehat{\mu}$ of the momentum map is a closed one, i.e. $d\widehat{\mu}=0$. If we consider a solution $\Gamma\colon U\subset M\rightarrow W_{L\eta}$ of $(\pi\colon E\rightarrow M,L\eta)$, then for each $\xi\in\lag$ we have
\[
d(\Gamma^*J_\xi)=\Gamma^*(dJ_\xi)=\Gamma^*(-\xi_{W_{L\eta}}\lrcorner d \lambda_{L\eta})=0,
\]
and therefore the momentum is conserved along solutions. Thus, we obtain Noether's theorem in this setting:
\begin{proposition} The momentum map $J$ is conserved along solutions of $(W_{L\eta},\lambda_{L\eta})$.
\end{proposition}
We might then restrict our attention to solutions which lie in the level set of a fixed (and closed) value of the momentum $\widehat{\mu}$. 

A momentum map is $Ad^*$-\emph{equivariant} if it satisfies
\[
 \langle J(g\alpha),Ad_{g^{-1}}\xi \rangle=g \langle J(\alpha),\xi \rangle. 
\]
Note that this is an equivariance condition for the natural action of $G$ on the spaces $W_{L\eta}$ and $\Lambda^{m-1} W_{L\eta} \otimes \lag^*$, where $G$ acts on $\lag^*$ by $g\mu=Ad_{g^{-1}}^*\mu$.  The construction of a momentum map for the action on $W_{L\eta}$ is standard~\cite{GotayIM_Momentum}:

\begin{lemma}\label{lem:mmap} The map $J\colon W_{L\eta}\to \Lambda^{m-1} W_{L\eta} \otimes \lag^*$ defined by
\[
\langle J(\alpha),\xi \rangle=\xi_{W_{L\eta}}(\alpha) \lrcorner \left.\lambda_{L\eta}\right|_\alpha,
\]
for each $\xi\in\lag$, is an $Ad^*$-equivariant momentum map for the $G$-action on $W_{L\eta}$.
\end{lemma}

We use the notation $J$ to denote the specific momentum map defined in Lemma~\ref{lem:mmap}. We will show that such an admissible momentum $\widehat{\mu}$ is $(\pi_1\circ \pi_{L\eta})$-basic, and can be though of as an $(m-1)$-form $\mu_M$ on $M$. From now on, we assume that $(\pi_1\circ\pi_{L\eta})$ has connected fibers.

\begin{lemma}\label{lem:momentumisbasic} Let $\widehat{\mu}\in\Omega^{m-1}(W_{L\eta},\g^*)$ be a closed $\g^*$-valued form on $W_{L\eta}$ in the image of $J$. Then there exists $\mu_M\in\Omega^{m-1}(M,\g^*)$ such that $\widehat{\mu}=(\pi_1\circ \pi_{L\eta})^*\,\mu_M$. In particular, $\widehat{\mu}$ is 1-horizontal.
\end{lemma}
\begin{proof} Let $\widehat{\mu}\mid_\alpha=J(\alpha)$. From the definition of the Cartan form $\lambda_{L\eta}$, if $v \in V_\alpha(\pi_1\circ \pi_{L\eta})$, then
\[
v\lrcorner J_\xi(\alpha)= v\lrcorner \big(\xi_{W_{L\eta}}\lrcorner \lambda_{L\eta}\big|_\alpha\big)=0,
\]
since $\alpha$ is 2-horizontal and $T\pi_{L\eta}(v),T\pi_{L\eta}(\xi_{W_{L\eta}})\in V\pi_1$. Therefore $\widehat{\mu}$ annihilates the vertical space of $(\pi_1\circ \pi_{L\eta})$. It remains to check that $\widehat{\mu}$ is constant on the fibers of $(\pi\circ\pi_{L\eta})$. This happens if, and only if,  $\cL_Z\widehat{\mu}=0$ for each $Z$ vertical w.r.t. $(\pi_1\circ \pi_{L\eta})$ (because $(\pi_1\circ \pi_{L\eta})$ has connected fibers). But this is immediate: using that $\widehat{\mu}$ is closed we have $\cL_Z\widehat{\mu}=d\left(Z\lrcorner\widehat{\mu}\right)=0$. 
\end{proof}
We will write $\mu=\pi_1^*\mu_M\in \Omega^{m-1}_1(J^1\pi,\g^*)$, and speak indistinctively of $\widehat{\mu}$ or $\mu$; note that $\mu$ is characterized by:
\begin{equation*}
\widehat{\mu}=\pi_{L\eta}^*\,\mu .
\end{equation*}

Given a closed $\mathfrak{g}^*$-valued form $\widehat{\mu}$ on $W_{L\eta}$, we  denote the corresponding momentum level set as follows:
\[
J^{-1}\left(\widehat{\mu}\right)=\left\{\alpha\in W_{L\eta}:\left.\widehat{\mu}\right|_{\alpha}=J\left(\alpha\right)\right\}.
\]
We will give an explicit description of the elements of this set in Lemma \ref{lem:momentum-set} below. It can be proved that this set is a particular instance of a \emph{momentum-type submanifold} of $W_{L\eta}$, as discussed in~\cite{AMN_17}. 

We will denote by $G_\mu$ the isotropy group of $\mu$, i.e. the subgroup of $G$ consisting of elements which leave $\mu$ invariant under the natural action on $\Omega^{m-1}(J^1\pi,\g^*)$:
\[
 G_\mu=\{g\in G: g\mu=\mu\}.
\]
It is easy to check that this subgroup coincides with the isotropy group of $\widehat{\mu}$ (defined analogously). Thus, $G_\mu$ acts on $W_{L\eta}$ and leaves $J^{-1}(\widehat{\mu})$ invariant.

\begin{remark} In the case of classical mechanics the configuration bundle is $Q\times \R\to \R$. We have $m=1$, and a momentum map is a map $J:W_L\to \lag^*$. A momentum value is identified with an element in $\lag^*$, as usual.
\end{remark}

There is an splitting of the contact bundle induced by the choice of a connection on the principal bundle $p^E_G\colon E\to E/G$, as we describe next. We denote by $\omega\in\Omega^1(E,\lag)$ the chosen connection and consider the following splitting of the cotangent bundle:
\[
T^*E=(p^E_G)^*\big(T^*(E/G)\big) \oplus (E\times \lag^*).
\]
The identification is obtained as follows:
\begin{align*}
(p^E_G)^*\big(T^*(E/G)\big) \oplus (E\times \lag^*)&\to T^*E,\\
(e,\widehat{\alpha}_{[e]},\sigma)&\mapsto \alpha_e=\widehat{\alpha}_{[e]}\circ T_ep_G^E+\langle\sigma,\omega(\cdot)\rangle.
\end{align*}
Accordingly, we have an splitting of contact bundle~\eqref{eq:ContactFields}
\begin{equation}\label{eq:splittingcontact_LFT}
I^m_{{\rm con},2}=\widetilde{I^m_{{\rm con},2}}\oplus I^m_{\lag^*,2},
\end{equation}
with
\begin{align*}
  \left.\widetilde{I^m_{{\rm con},2}}\right|_{j_x^1s}&=\mathcal{L}\,\Big\{\widehat{\alpha}_{\left[s\left(x\right)\right]}\circ T_{s\left(x\right)}p_G^E\circ(T_{j_x^1s}\pi_{10}-T_xs\circ T_{j_x^1s}\pi_1)\wedge\beta:\\
  &\hskip10em\widehat{\alpha}_{\left[s\left(x\right)\right]}\in T_{\left[s\left(x\right)\right]}^*\left(E/G\right),\beta\in\left(\Lambda^{m-1}_1J^1\pi\right)_{j_x^1s}\Big\},\\
  \left.I_{\g^*,2}^m\right|_{j_x^1s}&=\left\{\langle\sigma\stackrel{\wedge}{,}\omega\circ(T_{j_x^1s}\pi_{10}-T_xs\circ T_{j_x^1s}\pi_1)\rangle:\sigma\in\left(\Lambda^{m-1}_1J^1\pi\otimes\g^*\right)_{j_x^1s}\right\}.
\end{align*}
Here $\langle\cdot\stackrel{\wedge}{,} \cdot\rangle$ denotes the natural contraction. For simple tensors, writing $\alpha_1\otimes\nu$ for a $\lag^*$-valued form ($\nu\in\lag^*$) and $\alpha_2\otimes \eta$ for a $\lag$-valued form ($\eta\in\lag$), we have $\langle\alpha_1\otimes\nu \stackrel{\wedge}{,} \alpha_2\otimes\eta \rangle=\langle\nu,\eta\rangle \alpha_1\wedge\alpha_2$.

If $s\colon U\to E$ is a section, we can define a reduced section $[s]_G\colon U\to E/G$, whose value at $x\in M$ is simply $[s(x)]_G$. We also recall that a point in the fiber of $x\in M$ of the vector bundle $\Lin(\overline{\pi}^*TM,\widetilde{\lag})\simeq \overline{\pi}^*(T^*M)\otimes \widetilde{\lag}$ represents a linear map from $T_xM$ to $\lag$, where $\widetilde{\lag}$ is the adjoint bundle associated to the principal bundle $p_G^E:E\rightarrow E/G$. Sections of this bundle might be identified with linear bundle maps over the identity from $\overline{\pi}^*TM$ to $\widetilde{\lag}$. An element in $\Lin(\overline{\pi}^*TM,\widetilde{\lag})$ is of the form $[e,\widehat{\xi}]_G$, where $\overline{\pi}(e)\in E/G$ and $\widehat{\xi}\colon T_{\overline{\pi}(e)}M\to \lag$ is a linear map.

\begin{proposition} The map
  \begin{align*}
    \Upsilon_{\omega}\colon J^1\pi&\longto \left(p_G^E\right)^*\left(J^1\overline{\pi}\times_{E/G}\Lin\left(\overline{\pi}^*TM,\widetilde{\lag}\right)\right),\\
    j_x^1s&\longmapsto \left(s\left(x\right),j^1_x\left[s\right]_G,\left[s\left(x\right),\omega\circ T_xs\right]_G\right).
  \end{align*}
is a bundle isomorphism.
\end{proposition}

\begin{proof} We will construct explicitly the inverse $\Upsilon_{\omega}^{-1}$. We regard a 1-jet $\sigma$ of $\pi$ as a splitting of the sequence
\[
0\longto V_eE{\longto} T_eE\overset{\pi_*}{\longto} T_xM \longto 0,
\]
i.e. as a map $\sigma\colon T_xM\to T_eE$ with $\pi_*\circ\sigma=id_{T_xM }$, with $\pi_*=T\pi$. We need to define such an splitting starting from an splitting $\overline{\sigma}$ of the sequence
\[
0\longto V_{[e]}(E/G){\longto} T_{[e]}(E/G) \overset{\overline{\pi}_*}{\longto} T_xM \longto 0.
\]
The connection $\omega$ determines, via its horizontal lift $H_{\omega}\colon T_{[e]}(E/G)\to T_eE$, a splitting of the sequence
\[
0\longto V_eE{\longto} T_eE\overset{(p_G^E)_*}{\longto} T_{[e]}(E/G) \longto 0.
\]
Finally, we also have a linear map $\widehat{\xi}:T_xM\to \lag$. Then $\sigma=\Upsilon_{\omega}^{-1}\big(e,\overline{\sigma},\widehat{\xi}\,\big)$ is the splitting $\sigma\colon T_xM\to T_eE$ given by
\[
\sigma(v_x)=(H_\omega\circ\overline{\sigma})(v_x)+\big(\widehat{\xi}(v_x)\big)_Q(e).
\]
It is clear from the definition that $\sigma$ is the inverse of $\Upsilon_{\omega}$.
\end{proof}
The map $\Upsilon_\omega$ enjoys a useful property: under this identification, the action of $G$ on $J^1\pi$ is simply
\[
g\cdot\big(e,j_x^1\overline{s},[e,\widehat{\xi}]_G\big)=\big(g\cdot e,j_x^1\overline{s},[e,\widehat{\xi}]_G\big).
\]
This is a direct consequence of the equivariance of the principal connection $\omega$. As a result, we get the following corollary.
\begin{corollary}\label{cor:identification}
There is an identification
  \[
  J^1\pi/G_\mu\simeq  J^1\overline{\pi}\times_{E/G}E/G_\mu\times_{E/G} \Lin\left(\overline{\pi}^*TM,\widetilde{\lag}\right).
  \]
\end{corollary}

\begin{remark} The assumption of a global Lie group action on the configuration bundle adopted in this paper is standard in the literature of Lagrangian field theory reduction, see e.g.~\cite{Castrillon1,Castrillon2,GbRat,LagPoincare_JGP}. A recent approach to Routh reduction in the mechanical case~\cite{Geometryrouth} assumes only  the weaker notion of an infinitesimal symmetry, encoded by a vector field $X$ on the manifold, to perform reduction. It should be noted, however, that~\cite{Geometryrouth} deals only with the case of a single vector field (corresponding to the case of a cyclic variable). 
\end{remark}

\section{Routh reduction for Lagrangian field theories}\label{sec:RouthLFT}

We will describe an approach to Routh reduction for a LFT $(\pi\colon E\rightarrow M,L\eta)$ which is similar to the one discussed in~\cite{Capriotti_Routh} for the mechanical case. First, we need a definition of the Routhian in field theory.

\subsection{The Routhian in field theory}

We will consider solutions which have a prescribed value of the momentum map $\widehat{\mu}$. As noted above, $\widehat{\mu}$ is assumed to be closed, and this implies (Lemma~\ref{lem:momentumisbasic}) that it is of the form $\widehat{\mu}=\pi_{L\eta}^*\,\mu$ for some closed $\mu\in\Omega^{m-1}_1(J^1\pi,\g^*)$.  Let us denote
\[
W_{L\eta}^\mu=J^{-1}\left(\widehat{\mu}\right)
\]
the level set of $\widehat{\mu}$. We will denote by $\lambda^\mu_{L\eta}$ its canonical $m$-form (the pullback of $\lambda_{L\eta}$ by the inclusion) and simply write $\pi_{L\eta}$ for the projection onto $M$. 

Recall that we have an splitting of the contact bundle induced by the connection~\eqref{eq:splittingcontact_LFT}. The following lemma shows that the momentum fixes the vertical component of the forms in the contact structure. It will be convenient to use the following notation: 
\[
\varepsilon=(-1)^{\dim M -1}. 
\]
\begin{lemma}\label{lem:momentum-set} Any simple element $\rho\in W_{L\eta}^\mu$ such that $\pi_{L\eta}(\rho)=j_x^1s$ can be written as
\begin{multline*}
    \rho=L\left(j_x^1s\right)\eta+\widehat{\alpha}_{\left[s\left(x\right)\right]_G}\circ T_{s\left(x\right)}p_G^E\circ(T_{j_x^1s}\pi_{10}-T_xs\circ T_{j_x^1s}\pi_1)\wedge\beta+\\
    +\varepsilon\, \langle\mu \stackrel{\wedge}{,}\left.\omega\right|_{s\left(x\right)}\circ(T_{j_x^1s}\pi_{10}-T_xs\circ T_{j_x^1s}\pi_1)\rangle
\end{multline*}
for some $\widehat{\alpha}_{[s(x)]_G}\in T_{[s(x)]_G}^*(E/G)$ and $\beta\in(\Lambda^{m-1}_1J^1\pi)_{j_x^1s}$.
\end{lemma}
\begin{proof}
The form $\rho\in W_{L\eta}^\mu\subset\Lambda^m_2 J^1\pi$ can be written as
\begin{multline*}
    \rho=L\left(j_x^1s\right)\eta+\widehat{\alpha}_{\left[s\left(x\right)\right]_G}\circ T_{s\left(x\right)}p_G^E\circ(T_{j_x^1s}\pi_{10}-T_xs\circ T_{j_x^1s}\pi_1)\wedge\beta+\\
    +\langle\sigma \stackrel{\wedge}{,}\left.\omega\right|_{s\left(x\right)}\circ(T_{j_x^1s}\pi_{10}-T_xs\circ T_{j_x^1s}\pi_1)\rangle,
\end{multline*}
with all the terms as in the statement of the lemma, and $\sigma\in(\Lambda^{m-1}_1J^1\pi\otimes\lag^*)_{j_x^1s}$. If $\xi\in\lag$, the infinitesimal generator $\xi_{W_{L\eta}}$ satisfies $T\pi_{L\eta}(\xi_{W_{L\eta}})=\xi_{J^1\pi}$ (the generator of the prolonged action), which is $\pi_1$-vertical. In particular $\xi_{J^1\pi}(j^1_xs)\lrcorner\eta_{j^1_xs}=0$, and $\xi_{J^1\pi}(j^1_xs)\lrcorner \left.\beta\right|_{j^1_xs}=0$. Moreover, $\xi_{J^1\pi}$  satisfies $T_{j_x^1s}\pi_{10}\big((\xi_{J^1\pi})(j_x^1s)\big)=\xi_{E}(s(x))$, and therefore $(T_{j_x^1s}\pi_{10}-T_xs\circ T_{j_x^1s}\pi_1)\big(\xi_{J^1\pi}(j^1_xs)\big)=\xi_{E}(s(x))$, which is $p^E_G$ vertical. Finally, from the definition of the connection form, we have $\omega(\xi_E)=\xi$. All together, this means that
\begin{align*}
\langle J(\rho),\xi \rangle=\pi_{L\eta}^*(\xi_{J^1\pi}\lrcorner\rho)&=\langle\pi_{L\eta}^*\sigma,\varepsilon\, \xi_{J^1\pi}\lrcorner \omega\circ(T\pi_{10}-Ts\circ T_{j_x^1s}\pi_1) \rangle\\
&=\varepsilon\,\langle \pi_{L\eta}^*\sigma,\xi \rangle.
\end{align*}
Imposing $J(\rho)=\widehat{\mu}$ and writing $\widehat{\mu}=\pi_{L\eta}^*\,\mu$, the claim follows.
\end{proof}
For arbitrary elements $\rho\in W^\mu_{L\eta}$, we get a similar result, but with the second term of the form
\[
\sum_i  \widehat{\alpha}_{i}\circ T_{s\left(x\right)}p_G^E\circ(T_{j_x^1s}\pi_{10}-T_xs\circ T_{j_x^1s}\pi_1)\wedge\beta_i.
\]
\begin{remark} Note that Lemma~\ref{lem:momentum-set} implies that $W^\mu_{L\eta}\to M$ is an affine bundle.
\end{remark}

This suggests us to define a \emph{Routhian density} $\mathcal{R_\mu}\in\Omega^m_1(J^1\pi)$ as follows:
\[
\mathcal{R_\mu}\left(j_x^1s\right)=L\left(j_x^1s\right)\eta_x - \varepsilon \big\langle\left.\mu\right|_{j_x^1s}\stackrel{\wedge}{,}\left.\omega\right|_{s\left(x\right)}\circ T_xs\circ T_{j_x^1s}\pi_1\big\rangle,\qquad j_x^1s\in J^1\pi.
\]

\smallskip
\begin{proposition} The form $\mathcal{R_\mu}$ is a Lagrangian density. It is $G_\mu$-invariant.
\end{proposition}
\begin{proof} Recall that a Lagrangian density is a $\pi_1$-semibasic form on $J^1\pi$. It suffices to check that the second term is $\pi_1$-semibasic, but this is clear since it annihilates $V\pi_1$. For the $G_\mu$-invariance, note that $\mu$ is invariant by definition of $G_\mu$.
\end{proof}
We can naturally define the \emph{Routhian} to be the function $R_\mu\in C^\infty(J^1\pi)$ such that $\mathcal{R_\mu}=R_\mu\eta$. We will see later that, just like in the mechanical case, this function plays the role of the Lagrangian for the reduced system.

We write $\overline{p}: E/{G_\mu}\times{\rm Lin}{\left(\overline{\pi}^*TM,\widetilde{\lag}\right)}\to E/G$ for the obvious projection. In particular, one can consider the map:
\begin{align*}
  q:J^1\left(\overline{\pi}\circ\overline{p}\right)&\longto J^1\overline{\pi}\times E/{G_\mu}\times {\rm Lin}{\left(\overline{\pi}^*TM,\widetilde{\g}\right)},\\
  j_x^1\sigma&\longmapsto\left(j_x^1\left(\overline{p}\circ\sigma\right),\sigma(x)\right).
\end{align*}
Using the connection $\omega$, we have maps fitting in the following diagram:
\begin{equation*}
  \begin{tikzcd}[ampersand replacement=\&, column sep=1.2cm, row
    sep=.9cm]
    {J^1\pi} \arrow[swap]{d}{p_{G_\mu}^{J^1\pi}} \arrow{r}{f_\omega} \&
    {E/G_\mu\times{\rm Lin}{\left(\overline{\pi}^*TM,\widetilde{\lag}\right)}} \&
    {J^1\left(\overline{\pi}\circ\overline{p}\right)}
    \arrow[swap]{l}{\left(\overline{\pi}\circ\overline{p}\right)_{10}}
    \arrow{d}{q}
    \\
    {J^1\pi/G_\mu}
    \arrow[swap]{rr}{g_\omega}
    \&
    \&
    {J^1\overline{\pi}\times E/{G_\mu}\times{\rm Lin}{\left(\overline{\pi}^*TM,\widetilde{\lag}\right)}}
  \end{tikzcd}
\end{equation*}
The definitions are as follows:
\begin{align*}
f_\omega:J^1\pi& \longto E/{G_\mu}\times{\rm Lin}{\left(\overline{\pi}^*TM,\widetilde{\g}\right)},\\
j_x^1s&\longmapsto\big(\left[s\left(x\right)\right]_{G_\mu},\left[s\left(x\right),\omega\circ T_xs\right]_G\big).\\
 g_\omega:J^1\pi/{G_\mu}& \longto J^1\overline{\pi}\times E/{G_\mu}\times{\rm Lin}{\left(\overline{\pi}^*TM,\widetilde{\g}\right)},\\
\left[j_x^1s\right]_{G_\mu} &\longmapsto\big(j_x^1\left(p_G^E\circ s\right),\left[s\left(x\right)\right]_{G_\mu},\left[s\left(x\right),\omega\circ T_xs\right]_G\big).
\end{align*}
The map $g_\omega$ is the identification from Corollary~\ref{cor:identification}. Since the Routhian density $\mathcal{R}_\mu$ is invariant under $G_\mu$, it defines a reduced density on $J^1\pi/G_\mu$ which, under the identification $g_\omega$ can be seen as a density on $J^1\overline{\pi}\times E/G_\mu\times\mathop{\rm Lin}{\left(\overline{\pi}^*TM,\widetilde{\g}\right)}$. We will denote it by $\overline{\mathcal{R}}_\mu$:
\[
 \big(g_\omega\circ p_{G_\mu}^{J^1\pi}\big)^*\overline{\mathcal{R}}_\mu=\mathcal{R}_\mu,\qquad \overline{\mathcal{R}}_\mu \in \Omega_1^m\big(J^1\overline{\pi}\times E/G_\mu\times\mathop{\rm Lin}{\left(\overline{\pi}^*TM,\widetilde{\g}\right)}\big).
\]
Likewise, the Routhian $R_\mu$ defines a reduced function $\overline{R}_\mu$. Note that $\overline{\mathcal{R}}_\mu=\overline{R}_\mu\eta$. We will also call $\overline{\mathcal{R}}_\mu$ the Routhian density and $\overline{R}_\mu$ the Routhian.

\subsection{Some technical results}

This section contains some technical lemmas which will be used later to obtain the main results on the Routh reduction of Lagrangian field theories. We have shown that the extremals of a LFT with prescribed momentum $\mu$ are encoded in the affine subbundle $W^\mu_{L\eta}$. To relate this affine subbundle with the affine subbundle obtained from the reduced Lagrangian density $q^*\overline{\mathcal{R}}_\mu$ on $J^1(\overline{\pi}\circ\overline{p})$ we will make use of the following pullback bundle:
\[
F_\omega=f_\omega^*\left(J^1\left(\overline{\pi}\circ\overline{p}\right)\right).
\]
It fits into the following commutative diagram:
\begin{equation}\label{dia:Pullback_bundle}
  \begin{tikzcd}[ampersand replacement=\&, column sep=.7cm, row
    sep=.9cm]
    {{F_\omega}=
      f_\omega^*\left(J^1\left(\overline{\pi}\circ\overline{p}\right)\right)}
    \arrow[dashed]{r}{\text{pr}_2^\omega}
    \arrow[dashed,swap]{d}{\text{pr}_1^\omega}
    \&
    {J^1\left(\overline{\pi}\circ\overline{p}\right)}
    \arrow{d}{\left(\overline{\pi}\circ\overline{p}\right)_{10}}
    \arrow[bend left=60, end anchor=north east]{dd}{q}
    \\
    {J^1\pi}
    \arrow{r}{f_\omega}
    \arrow[swap]{d}{p_{G_\mu}^{J^1\pi}}
    \&
    E/G_\mu\times\mathop{\rm Lin}{\left(\overline{\pi}^*TM,\widetilde{\g}\right)}
    \\
    J^1\pi/G_\mu
    \arrow[swap]{r}{g_\omega}
    \&
    J^1\overline{\pi}\times E/G_\mu\times\mathop{\rm Lin}{\left(\overline{\pi}^*TM,\widetilde{\g}\right)}
    \arrow[swap]{u}{\text{pr}_{23}}
  \end{tikzcd}
\end{equation}
The maps $\text{pr}_1^\omega$ and $\text{pr}_2^\omega$ are the canonical projections of the pullback bundle $F_\omega$ onto $J^1\pi$ and $J^1(\overline{\pi}\circ\overline{p})$, respectively. We consider the affine subbundles:
\begin{align*}
  W_{L\eta}^\mu&=J^{-1}\left(\mu\right)\subset\Lambda^m_2J^1\pi,\\
  W^0_{q^*\overline{\cal R}_\mu}&=q^*{\overline{\cal R}_\mu}+\widehat{I}^m_{\text{con},2}\subset\Lambda^m_2J^1(\overline{\pi}\circ\overline{p}),
\end{align*}
where
\begin{align*}
  \left.\widehat{I}^m_{\text{con},2}\right|_{j_x^1\sigma}=\Big\{[\widehat{\alpha}\circ&(T_{j_x^1\left(\overline{p}\circ\sigma\right)} \overline{\pi}_{10}-T_x\left(\overline{p}\circ\sigma\right)\circ T_{j_x^1\left(\overline{p}\circ\sigma\right)}\overline{\pi}_1)\circ T_{j_x^1\sigma}j^1\overline{p}]\wedge\beta:\\
 & \widehat{\alpha}\in T^*_{\overline{p}\left(\sigma\left(x\right)\right)}\left(E/G\right), \beta\in\left.\left(\Lambda^{m-1}_1J^1\left(\overline{\pi}\circ\overline{p}\right)\right)\right|_{j_x^1\sigma}\Big\}\subset \Lambda^m_2J^1(\overline{\pi}\circ\overline{p})
\end{align*}
is essentially the pullback of the contact subbundle of $J^1\overline{\pi}$ to $J^1(\overline{\pi}\circ\overline{p})$. Note that $\widehat{I}^m_{\text{con},2}$ is a subbundle of the contact subbundle of $J^1(\overline{\pi}\circ\overline{p})$. Thus we have a inclusion 
\[
 W^0_{q^*\overline{\cal R}_\mu}\subset W_{q^*\overline{\cal R}_\mu},
\]
where $W_{q^*\overline{\cal R}_\mu}$ is the affine translation  of the contact subbundle of $J^1(\overline{\pi}\circ\overline{p})$ by $q^*\overline{\cal R}_\mu$ (the notation is consistent with~\eqref{eq:affinesubbundle}). Finally, we construct the following subbundles of $\Lambda^m_2\left(F_\omega\right)$: 
\begin{align*}
  \left(\text{pr}_1^\omega\right)^*(W_{L\eta}^\mu)\big|_{\rho}& =\left\{\alpha(T_{\rho}\text{pr}_1^\omega(\cdot),\dots,T_{\rho}\text{pr}_1^\omega(\cdot)) \in\left.\Lambda^m_2\left(F_\omega\right)\right|_\rho:\alpha\in\left.W_{L\eta}^\mu\right|_{j_x^1s}\right\},\\
  (\text{pr}_2^\omega)^*(W^0_{q^*\overline{\cal R}_\mu})\big|_{\rho}&= \Big\{\kappa(T_{\rho}\text{pr}_2^\omega(\cdot),\dots,T_{\rho}\text{pr}_2^\omega(\cdot)) \in\left.\Lambda^m_2\left(F_\omega\right)\right|_\rho:\kappa\in\left.W^0_{q^*\overline{\mathcal{R}}_\mu}\right|_{j_x^1\sigma}\Big\},
\end{align*}
at a point $\rho=(j_x^1s,j_x^1\sigma)\in F_\omega$ which is such that $\sigma(x)=([s(x)]_{G_\mu},[s(x),\omega\circ T_xs]_G)$. Note that these bundles are obtained via pullback -using the projections  $\text{pr}_1^\omega$, $\text{pr}_2^\omega$- of the corresponding affine subbundles. We will write
\begin{align*}
  \Pi_{L\eta}\colon \left(\text{pr}_1^\omega\right)^*W_{L\eta}^\mu&\longto W_{L\eta}^\mu,\; (\text{pr}_1^\omega)^*\alpha\mapsto\alpha,\\
  \Pi_{q^*\overline{\mathcal{R}}_\mu}\colon \left(\text{pr}_2^\omega\right)^*W^0_{q^*\overline{\mathcal{R}}_\mu}&\longto W^0_{q^*\overline{\mathcal{R}}_\mu},\; (\text{pr}_2^\omega)^*\kappa \mapsto\kappa,
\end{align*}
for the projections. The following diagram summarizes the situation:
\begin{equation}\label{dia:aux1}
\begin{tikzcd}[column sep= 1.2cm, row sep = 1cm]
\left(\text{pr}_1^\omega\right)^*W_{L\eta}^\mu \arrow[dr,swap,"\Pi_{L\eta}"]\arrow[r,hook,"i_{L\eta}^\omega"] & \Lambda^m_2 F_\omega\arrow[r,"\psi"] & F_\omega\arrow[d,"\text{pr}_1^\omega"]  \\
& W_{L\eta}^\mu \arrow[r,swap,"\pi_{L\eta}"] & J^1\pi
\end{tikzcd}
\end{equation}
where $\psi\colon \Lambda^m_2 F_\omega\to F_\omega$ is the canonical projection and $i_{L\eta}^\omega$ is the natural inclusion. There is an equivalent diagram for $({\rm pr}_2^\omega)^*W^0_{q^*\overline{\mathcal{R}}_\mu}$.
\begin{lemma}\label{lem:tech1} Let $\lambda_{L\eta}'\in\Omega^m\big(({\rm pr}_1^\omega)^*W_{L\eta}^\mu\big)$ be the pullback of the canonical $m$-form $\lambda'$ on $\Lambda^m_2 F_\omega$ to $\left({\rm pr}_1^\omega\right)^*W_{L\eta}^\mu$. Then
\[
\Pi_{L\eta}^*\lambda_{L\eta}=\lambda_{L\eta}'.
\]
\end{lemma}
\begin{proof} Consider $m$ tangent vectors  vectors  $v_1,\dots,v_n\in T_\alpha({\rm pr}_1^\omega)^*W_{L\eta}^\mu$. To simplify the notation, we will write $i_{L\eta}^\omega(\alpha)=\alpha$, and then vectors at an element $\alpha\in ({\rm pr}_1^\omega)^*W_{L\eta}^\mu$ are also seen as vectors at $\alpha\in\Lambda^m F_\omega$. By definition, we have
\[
\left.\lambda_{L\eta}'\right|_\alpha\big(v_1,\dots,v_m\big)=\alpha\big(T_\alpha \psi (v_1),\dots,T_\alpha \psi (v_m)\big).
\]
But $\alpha$ is of the form $\alpha=({\rm pr}_1^\omega)^*\beta$ for some $\beta\in W^\mu_{L\eta}$ (namely, $\Pi_{L\eta}(\alpha)=\beta$), hence
\[
\left.\lambda_{L\eta}'\right|_\alpha\big(v_1,\dots,v_m\big)=\beta\big(T_{\psi(\alpha)}{\rm pr}_1^\omega\circ T_\alpha \psi (v_1),\dots,T_{\psi(\alpha)}{\rm pr}_1^\omega\circ T_\alpha \psi (v_m)\big).
\]
On the other hand, using that $\Pi_{L\eta}(\alpha)=\beta$, we have
\begin{align*}
\left.(\Pi_{L\eta}^*\lambda_{L\eta})\right|_\alpha\big(v_1,\dots,v_m\big)& =\left.\lambda_{L\eta}\right|_\beta\big(T_\alpha\Pi_{L\eta}(v_1),\dots,T_\alpha\Pi_{L\eta}(v_m)\big) \\
&= \beta(T_{\beta}\pi_{L\eta}\circ T_\alpha\Pi_{L\eta}(v_1),\dots,T_{\beta}\pi_{L\eta}\circ T_\alpha\Pi_{L\eta}(v_m)\big),
\end{align*}
which, looking at Diagram~\eqref{dia:aux1}, agrees with $\lambda_{L\eta}'|_\alpha(v_1,\dots,v_m)$.
\end{proof}
In the same way one proves the following. If $\lambda_{q^*\overline{\mathcal{R}}_\mu}'\in\Omega^m(({\rm pr}_2^\omega)^*W^0_{q^*\overline{\mathcal{R}}_\mu})$ is the pullback of the canonical $m$-form $\lambda'$ on $\Lambda^m_2(F_\omega)$ to $({\rm pr}_2^\omega)^*W^0_{q^*\overline{\mathcal{R}}_\mu}$, then
\[
\Pi_{q^*\overline{\mathcal{R}}_\mu}^*\lambda_{q^*\overline{\mathcal{R}}_\mu}=\lambda_{q^*\overline{\mathcal{R}}_\mu}'.
\]

The following lemma shows that solutions of $(\pi\colon E\rightarrow M,L\eta)$ with momentum $\mu$, thought of as sections $\Gamma$ of $W^\mu_{L\eta}\to M$, can be identified with sections $\widehat{\Gamma}$ of $({\rm pr}_1^\omega)^*W^\mu_{L\eta}\to M$.

\begin{lemma}\label{lem:liftingGamma}  Let $\Gamma\colon M\rightarrow W_{L\eta}^\mu$ be a section such that
\[
\Gamma^*\big(Z\lrcorner d\lambda^\mu_{L\eta}\big)=0,\quad \text{for all } Z\in\mathfrak{X}^{V(\pi_1\circ\pi_{L\eta})}\left(W_{L\eta}\right),
\]
then there exists a section $\widehat{\Gamma}\colon M\rightarrow\left({\rm pr}_1^\omega\right)^*W_{L\eta}^\mu$ such that $\Pi_{L\eta}\circ\widehat{\Gamma}=\Gamma$ and
\[
\widehat{\Gamma}^*\left(Z'\lrcorner d\lambda_{L\eta}'\right)=0, \quad \text{for all } Z'\in\mathfrak{X}^{V(\pi_1\circ\pi_{L\eta}\circ\Pi_{L\eta})}\big(\left({\rm pr}_1^\omega\right)^*W_{L\eta}^\mu\big).
\]
Conversely, if $\widehat{\Gamma}\colon M\to\left({\rm pr}_1^\omega\right)^*W_{L\eta}^\mu$ is a section such that
\[
\widehat{\Gamma}^*\left(Z'\lrcorner d\lambda_{L\eta}'\right)=0, \quad \text{for all } Z'\in\mathfrak{X}^{V(\pi_1\circ\pi_{L\eta}\circ\Pi_{L\eta})}\big(\left({\rm pr}_1^\omega\right)^*W_{L\eta}^\mu\big),
\]
then the section $\Gamma=\Pi_{L\eta}\circ \widehat{\Gamma}\colon M\to W_{L\eta}^\mu$ satisfies
\[
\Gamma^*\big(Z\lrcorner d\lambda^\mu_{L\eta}\big)=0,\quad \text{for all } Z\in\mathfrak{X}^{V(\pi_1\circ\pi_{L\eta})}\left(W_{L\eta}\right).
\]
\end{lemma}
\begin{proof} The situation is illustrated in the following diagram:
\begin{equation*}
  \begin{tikzcd}[column sep=1.3cm, row sep=1.5cm]
   \left({\rm pr}_1^\omega\right)^*W_{L\eta}^\mu \arrow[rr,"\Pi_{L\eta}"]\arrow[d,"\tau_{L\eta}"] & &  W_{L\eta}^\mu\arrow[d,swap,"\pi_{L\eta}"]\\
  F_\omega \arrow[rr,"{\rm pr}^\omega_1"]\arrow[dr,swap] & &  J^1\pi \arrow[dl,"\pi_1"]\\
  & M\arrow[uur,dashed,swap,"\Gamma",bend right=65]\arrow[uul,dashed,"\widehat{\Gamma}",bend left=65] &  
  \end{tikzcd}
\end{equation*}

To define $\widehat{\Gamma}$, we start defining two sections:
\begin{align*}
\gamma_1&\colon M\to J^1\pi,  \quad \gamma_1=\pi_{L\eta}\circ\Gamma,\\
\gamma_2&\colon M\to E/{G_\mu}\times{\rm Lin}{\left(\overline{\pi}^*TM,\widetilde{\g}\right)},  \quad \gamma_2=f_\omega\circ\gamma_1.
\end{align*}
Note that if $\Gamma$ was constructed from the Euler-Lagrange equations as in Proposition~\ref{prop:FieldTheoryEqsWL}, $\gamma_1$ would be $j^1s$ for some $s\colon M\to E$. Using $\gamma_1$ and $\gamma_2$ we construct the section $\hat\gamma\colon M\to F_\omega$:
\[
\hat\gamma(x)=\left(\gamma_1(x),j^1\gamma_2(x)\right)\in F_\omega\subset J^1\pi\times J^1(\overline{\pi}\circ\overline{p}). 
\]
It is easy to check that it is well defined. Finally, $\widehat{\Gamma}\colon M\to\left({\rm pr}_1^\omega\right)^*W_{L\eta}^\mu$ is given by:
\[
\widehat{\Gamma}(x)=\Gamma(x)\circ T_{\hat\gamma(x)}{\rm pr}_1^\omega\in ({\rm pr}_1^\omega)^*W_{L\eta}^\mu\big|_{\hat\gamma(x)}. 
\]

Every $Z'\in\mathfrak{X}^{V(\pi_1\circ\pi_{L\eta}\circ\Pi_{L\eta})}\big(\left({\rm pr}_1^\omega\right)^*W_{L\eta}^\mu\big)$ can be written as
  \[
  Z'=f_1Z_1'+f_2Z_2'
  \]
  where $f_1,f_2$ are smooth functions on $(\text{pr}_1^\omega)^*W_{L\eta}^\mu$ and
  \[
  T\Pi_{L\eta}\circ Z_1'=Z_1\in\mathfrak{X}^{V(\pi_1\circ\pi_{L\eta})}\left(W_{L\eta}\right),\qquad T\Pi_{L\eta}\circ Z_2'=0.
  \]
Let $\widehat{\Gamma}:M\rightarrow\left(\text{pr}_1^\omega\right)^*W_{L\eta}^\mu$ be the section defined before. It satisfies $\Pi_{L\eta}\circ\widehat{\Gamma}=\Gamma$. By Lemma~\ref{lem:tech1} we have that
\begin{align*}
f_1Z_1'\lrcorner d\lambda_{L\eta}'= f_1Z_1'\lrcorner \Pi_{L\eta}^*d\lambda_{L\eta}=f_1 \Pi_{L\eta}^*(Z_1\lrcorner d\lambda_{L\eta}),
\end{align*}
and similarly $f_2Z_2'\lrcorner d\lambda_{L\eta}'=f_2 \Pi_{L\eta}^*(0\lrcorner d\lambda_{L\eta})=0$. Thus,
\begin{align*}
 \widehat{\Gamma}^*(Z'\lrcorner d\lambda_{L\eta}')=\widehat{\Gamma}^*(f_1 \Pi_{L\eta}^*(Z_1\lrcorner d\lambda_{L\eta}))=(f_1\circ \widehat{\Gamma})\; \Gamma^*(Z_1\lrcorner d\lambda_{L\eta})=0,
\end{align*}
as required. The converse is analogous.
\end{proof}
In the same way, one proves the following:
\begin{lemma}\label{lem:liftingGamma2} If $\widehat{\Gamma}\colon M\to({\rm pr}_2^\omega)^*W^0_{q^*\overline{\mathcal{R}}_\mu}$ is a section such that
\[
\widehat{\Gamma}^*\left(Z'\lrcorner d\lambda_{L\eta}'\right)=0, \quad \text{for all } Z'\in\mathfrak{X}^{V((\overline{\pi}\circ\overline{p})_1\circ\pi_{q^*\overline{\mathcal{R}}_\mu}\circ\Pi_{q^*\overline{\mathcal{R}}_\mu})}\big(({\rm pr}_2^\omega)^*W^0_{q^*\overline{\mathcal{R}}_\mu}\big),
\]
then the section $\Gamma=\Pi_{q^*\overline{\mathcal{R}}_\mu}\circ \widehat{\Gamma}\colon M\to W^0_{q^*\overline{\mathcal{R}}_\mu}$ satisfies
\[
\Gamma^*\big(Z\lrcorner d\lambda^0_{q^*\overline{\mathcal{R}}_\mu}\big)=0,\quad \text{for all } Z\in\mathfrak{X}^{V((\overline{\pi}\circ\overline{p})_1\circ\pi_{q^*\overline{\mathcal{R}}_\mu})}\big(W^0_{q^*\overline{\mathcal{R}}_\mu}\big).
\]
Here $\lambda^0_{q^*\overline{\mathcal{R}}_\mu}$ is the canonical $m$-form on  $W^0_{q^*\overline{\mathcal{R}}_\mu}$.
\end{lemma}

\begin{remark}\label{rem:1} Lemma~\ref{lem:liftingGamma} and Lemma~\ref{lem:liftingGamma2} also hold in the presence of a force term. The proof is similar. We point out that sections of $\Gamma\colon M\to W^0_{q^*\overline{\mathcal{R}}_\mu}$ in Lemma~\ref{lem:liftingGamma2} cannot in general be lifted to $({\rm pr}_2^\omega)^*W^0_{q^*\overline{\mathcal{R}}_\mu}$. We will discuss this in detail later in Section~\ref{sec:reconstruction} when we deal with reconstruction.
\end{remark}

We now discuss the force induced by the connection form. Consider the following 2-horizontal $m$-form $\omega_\mu\in \Omega^m_2(J^1\pi)$:
\[
\left.\omega_\mu\right|_{j_x^1s}=\varepsilon \big\langle\left.\mu\right|_x \stackrel{\wedge}{,} \left.\omega\right|_{s\left(x\right)}\circ T_{j_x^1s}\pi_{10}\big\rangle.
\]
One can show that $\omega_\mu$ is $G_\mu$-invariant as follows. The Lie group $G$ acts on $T(J^1\pi)$ by tangent lift (of the prolongation $j^1\phi_g$), which in particular implies equivariance of $T\pi_{10}$, i.e. $T\pi_{10}(gv)=gT\pi_{10}(v)$ for any $v\in T(J^1\pi)$. Using this observation and the equivariance of the connection, the invariance is immediate. Note that the same invariance holds for $d\omega_\mu$, and therefore there exists 
\[
\overline{\beta}_\mu\in\Omega^{m+1}_3\left(J^1\overline{\pi}\times E/G_\mu\times{\rm Lin}{(\overline{\pi}^*TM,\widetilde{\lag})}\right) 
\]
such that (see Diagram~\eqref{dia:Pullback_bundle}):
\[
\big(g_\omega\circ p_{G_\mu}^{J^1\pi}\big)^*\overline{\beta}_\mu=d\omega_\mu.
\]
The form $\overline{\beta}_\mu$ induces a force on $W^0_{q^*\overline{\mathcal{R}}_\mu}$, which we denote by $\beta^0_\mu\in \Omega^{m+1}_3(W^0_{q^*\overline{\mathcal{R}}_\mu})$, given by:
\[
\beta^0_\mu=\big(q\circ\pi_{q^*\overline{\mathcal{R}}_\mu}\big)^*\overline{\beta}_\mu.
\]
Finally, let us write $\widehat{\omega}_\mu=({\rm pr}_1^\omega)^*\omega_\mu$. 
\begin{lemma}\label{lem:GyroForcesRelation}
  The following holds:
  \[
  \Pi_{q^*\overline{\mathcal{R}}_\mu}^*\beta^0_\mu= d\big[\big(i^\omega_{q^*\overline{\mathcal{R}}_\mu}\big)^*\psi^*\widehat{\omega}_\mu\big].
  \]
\end{lemma}

\begin{proof}
The proof follows from diagram chasing in the following commutative diagram:
\begin{equation*}
\begin{tikzcd}[column sep=.8cm, row sep=1.5cm,every label/.append style = {font=\small}]
\Lambda^m_2\left(F_\omega\right) \arrow{r}{\psi}   &
F_\omega \arrow{r}{\text{pr}_1^\omega} &
J^1\pi\arrow{r}\arrow{r}{p_{G_\mu}^{J^1\pi}} &
J^1\pi/{G_\mu}\arrow{d}{g_\omega}
\\
({\rm pr}_2^\omega)^*W^0_{q^*\overline{\mathcal{R}}_\mu} \arrow[hook]{u}{i^\omega_{q^*\overline{\mathcal{R}}_\mu}}\arrow{r}[swap]{\Pi_{q^*\overline{\mathcal{R}}_\mu}} &
W^0_{q^*\overline{\mathcal{R}}_\mu} \arrow{r}[swap]{\pi_{q^*\overline{\mathcal{R}}_\mu}} &
J^1\left(\overline{\pi}\circ\overline{p}\right) \arrow{r}[swap]{q} &
J^1\overline{\pi}\times E/{G_\mu}\times{\rm Lin}{\left(\overline{\pi}^*TM,\widetilde{\g}\right)}
  \end{tikzcd}
\end{equation*}

Indeed, we have:
  \begin{align*}
    \Pi_{q^*\overline{\mathcal{R}}_\mu}^*\beta^0_\mu&=\big(q\circ\pi_{q^*\overline{\mathcal{R}}_\mu}\circ \Pi_{q^*\overline{\mathcal{R}}_\mu}\big)^*\overline{\beta}_\mu =\big(g_\omega\circ p_{G_\mu}^{J^1\pi}\circ\text{pr}_1^\omega\circ\psi\circ i^\omega_{q^*\overline{\mathcal{R}}_\mu}\big)^*\overline{\beta}_\mu\\
    &=\big(\psi\circ i^\omega_{q^*\overline{\mathcal{R}}_\mu}\big)^*(\text{pr}_1^\omega)^*\big(g_\omega\circ p_{G_\mu}^{J^1\pi}\big)^*\overline{\beta}_\mu =d\big[(i^\omega_{q^*\overline{\mathcal{R}}_\mu})^*\psi^*\widehat{\omega}_\mu\big],
  \end{align*}
  as required.
\end{proof}

We are ready to prove a key result which relates the subbundles $(\text{pr}_1^\omega)^*W_{L\eta}^\mu$ and $({\rm pr}_2^\omega)^*W^0_{q^*\overline{\mathcal{R}}_\mu}$ of $\Lambda^m_2 F_\omega$. Roughly speaking, they are related by an affine translation by means of the force $\widehat{\omega}_\mu$. More precisely, let us denote ${\rm t}_{\widehat{\omega}_\mu}\colon \Lambda^m_2 F_\omega\to \Lambda^m_2 F_\omega$ the map
\[
{\rm t}_{\widehat{\omega}_\mu}\left(\rho\right)=\rho+\left.\widehat{\omega}_\mu\right|_{\left(j_x^1s,j_x^1\sigma\right)},
\]
where $\left(j_x^1s,j_x^1\sigma\right)=\psi\left(\rho\right)$. Then the following holds:

\begin{proposition}\label{prop:WlWR_relation} With the notations above,
\[
{\rm t}_{\widehat{\omega}_\mu}\big(({\rm pr}_2^\omega)^*W^0_{q^*\overline{\mathcal{R}}_\mu}\big)=\left({\rm pr}_1^\omega\right)^*W_{L\eta}^\mu.
\]
\end{proposition}

\begin{proof} We will only show the inclusion ${\rm t}_{\widehat{\omega}_\mu}(({\rm pr}_2^\omega)^*W^0_{q^*\overline{\mathcal{R}}_\mu})\subset \left({\rm pr}_1^\omega\right)^*W_{L\eta}^\mu$. The converse is similar. 

Consider the following commutative diagram:
\begin{equation*}
\begin{tikzcd}[column sep=1.3cm, row sep=1.5cm]
F_\omega\arrow{r}{{\rm pr}_2^\omega}\arrow{d}[swap]{{\rm pr}_1^\omega} & J^1\left(\overline{\pi}\circ\overline{p}\right)\arrow{rr}{j^1\overline{p}}\arrow{d}{\left(\overline{\pi}\circ\overline{p}\right)_{10}} &&
      J^1\overline{\pi}\arrow{d}{\overline{\pi}_1}\arrow{dl}[swap]{\overline{\pi}_{10}}\\
J^1\pi\arrow{r}[swap]{f_\omega}&
E/G_{\mu}\times\mathop{\rm Lin}{\left(\overline{\pi}^*TM,\widetilde{\lag}\right)}\arrow{r}[swap]{\overline{p}} &
E/G     \arrow{r}[swap]{\overline{\pi}} &
M
\end{tikzcd}
\end{equation*}
Differentiating the relation
\[
\overline{\pi}_1\circ j^1\overline{p}\circ {\rm pr}_2^\omega=\overline{\pi}\circ \overline{p}\circ f_\omega\circ {\rm pr}_1^\omega=\pi_1\circ {\rm pr}_1^\omega,
\]
we find
\begin{equation}\label{eq:A}
 T_{j_x^1\left(\overline{p}\circ\sigma\right)}\overline{\pi}_1\circ T_{j_x^1\sigma}j^1\overline{p}\circ T_{\left(j_x^1s,j_x^1\sigma\right)}{\rm pr}_2^\omega = T_{j_x^1s}\pi_1\circ T_{\left(j_x^1s,j_x^1\sigma\right)}\text{pr}_1^\omega.
\end{equation}
In a similar way, from
\[
\overline{\pi}_{10}\circ j^1\overline{p}\circ {\rm pr}_2^\omega = \overline{p}\circ f_\omega\circ {\rm pr}_1^\omega = p^E_G\circ \pi_{10}\circ {\rm pr}_1^\omega,
\]
we get
\begin{equation}\label{eq:B}
T_{j_x^1\left(\overline{p}\circ\sigma\right)}\overline{\pi}_{10}\circ T_{j_x^1\sigma}j^1\overline{p}\circ T_{\left(j_x^1s,j_x^1\sigma\right)}{\rm pr}_2^\omega=T_{s\left(x\right)}p_G^E\circ T_{j_x^1s}\pi_{10}\circ T_{\left(j_x^1s,j_x^1\sigma\right)}\text{pr}_1^\omega.
\end{equation}

Since $(j_x^1s,j_x^1\sigma)\in F_\omega$, by definition we have $f_\omega(j_x^1s)=(\overline{\pi}\circ\overline{p})_{10}(j_x^1\sigma)=\sigma(x)$, and therefore
\[
\left(\overline{p}\circ\sigma\right)(x)=\left(\overline{p}\circ f_\omega\right)(j_x^1s).
\]
Recalling the relation between the maps above in the next diagram
\[
    \begin{tikzcd}[ampersand replacement=\&, column sep=1.3cm, row sep=1.5cm]
      J^1\pi
      \arrow{r}{f_\omega}
      \arrow{d}[swap]{\pi_{10}}
      \&
      {E/G_\mu \times\mathop{\text{Lin}}{\left(\overline{\pi}^*TM,\widetilde{\g}\right)}}
      \arrow{d}{\overline{p}}
      \\
      E
      \arrow{r}[swap]{p_G^E}
      \&
      E/G
    \end{tikzcd}
\]
we see that $p_G^E\left(s(x)\right)=\overline{p}\left(\sigma(x)\right)$, and thus
\begin{equation}\label{eq:C}
    T_x\left(\overline{p}\circ\sigma\right)=T_{s\left(x\right)}p_G^E\circ T_xs.
\end{equation}

Let us pick a simple element $\rho\in ({\rm pr}_2^\omega)^*W^0_{q^*\overline{\mathcal{R}}_\mu}$. It can be written as:
  \begin{multline*}
  \rho=\big\{\overline{\mathcal{R}}_\mu\left(q\left(j_x^1\sigma\right)\right)+\\
  +\left[\widehat{\alpha}\circ\left(T_{j_x^1\left(\overline{p}\circ\sigma\right)}\overline{\pi}_{10}- T_x\left(\overline{p}\circ\sigma\right)\circ T_{j_x^1\left(\overline{p}\circ\sigma\right)}\overline{\pi}_1\right)\circ T_{j_x^1\sigma}j^1\overline{p}\right]\wedge\beta\big\}\circ T_{\left(j_x^1s,j_x^1\sigma\right)}\text{pr}_2^\omega,
  \end{multline*}
for some $\widehat{\alpha}\in T_{\overline{p}\left(\sigma\left(x\right)\right)}^*\left(E/G\right)$ and $\beta\in\left.\left(\Lambda^{m-1}_1J^1\left(\overline{\pi}\circ\overline{p}\right)\right)\right|_{j_x^1\sigma}$. Taking into account~\eqref{eq:A}, \eqref{eq:B} and~\eqref{eq:C}, we have
\begin{multline*}
  \widehat{\alpha}\circ(T_{j_x^1\left(\overline{p}\circ\sigma\right)}\overline{\pi}_{10} - T_x\left(\overline{p}\circ\sigma\right)\circ T_{j_x^1\left(\overline{p}\circ\sigma\right)}\overline{\pi}_1)\circ T_{j_x^1\sigma}j^1\overline{p}\circ T_{\left(j_x^1s,j_x^1\sigma\right)}\text{pr}_2^\omega =\\
  =\widehat{\alpha}\circ T_{s\left(x\right)}p_G^E\circ(T_{j_x^1s}\pi_{10}-T_xs\circ T_{j_x^1s}\pi_1)\circ T_{\left(j_x^1s,j_x^1\sigma\right)}\text{pr}_1^\omega,
\end{multline*}
and then $\rho$ might as well be written as
\[
  \rho=\left\{\overline{\mathcal{R}}_\mu(q\left(j_x^1\sigma)\right) +
  [\widehat{\alpha}\circ T_{s\left(x\right)}p_G^E\circ(T_{j_x^1s}\pi_{10}-T_xs\circ T_{j_x^1s}\pi_1)]\wedge\alpha\right\}\circ T_{\left(j_x^1s,j_x^1\sigma\right)}\text{pr}_1^\omega,
\]
where $\alpha\in\left.\big(\Lambda^{m-1}_1J^1\pi\big)\right|_{j_x^1s}$ is chosen such that
\[
  \beta\circ T_{\left(j_x^1s,j_x^1\sigma\right)}\text{pr}_2^\omega=\alpha\circ T_{\left(j_x^1s,j_x^1\sigma\right)}\text{pr}_1^\omega.
\]
Note that, by definition,
\[
q(j_x^1\sigma)=\big(j_x^1(\overline{p}\circ\sigma),\sigma(x)\big)=
\big(j_x^1(p_G^E\circ s),f_\omega(s(x))\big) =
g_\omega\big(p_{G_\mu}^{J^1\pi}(j_x^1s)\big),
\]
and then $\overline{\mathcal{R}}_\mu(q\left(j_x^1\sigma)\right)=\mathcal{R}_\mu(j_x^1s)$. Thus, $\rho+ \left.\widehat{\omega}_\mu\right|_{\left(j_x^1s,j_x^1\sigma\right)}$ can be written as:
\begin{multline*}
  \rho+\left.\widehat{\omega}_\mu\right|_{\left(j_x^1s,j_x^1\sigma\right)} = \big\{ L(j^1_xs)\eta_x- \varepsilon \big\langle\left.\mu\right|_{j_x^1s}\stackrel{\wedge}{,}\left.\omega\right|_{s\left(x\right)}\circ T_xs\circ T_{j_x^1s}\pi_1\big\rangle + \\
+
  [\widehat{\alpha}\circ T_{s\left(x\right)}p_G^E\circ(T_{j_x^1s}\pi_{10}-T_xs\circ T_{j_x^1s}\pi_1)]\wedge\alpha\big\}\circ T_{\left(j_x^1s,j_x^1\sigma\right)}\text{pr}_1^\omega.
\end{multline*}
The reasoning is the same for an arbitrary (i.e. not necessarily simple) element $\rho\in ({\rm pr}_2^\omega)^*W^0_{q^*\overline{\mathcal{R}}_\mu}$. This concludes the proof.
\end{proof}

In order to apply this result to the reduction of field equations of motion, it will be necessary to take into account the following general fact. Roughly speaking, the next result states that the affine translation by $\widehat{\omega}_\mu$ will give rise to a force term related to its exterior differential $d\widehat{\omega}_\mu$.

\begin{lemma} \label{lem:TranslationSpaceForms} Let $P$ be a manifold and $\alpha\in\Omega^m\left(P\right)$ a $m$-form on $P$. Consider the affine translation ${\rm t}_{\alpha}\colon \Lambda^m P\to \Lambda^m P$ induced by $\alpha$, i.e.
  \[
  {\rm t}_{\alpha}(\beta)=\beta+\alpha_p,
  \]
where $p=\pi_P(\beta)$ ($\pi_P\colon \Lambda^m P\to P$ is the projection). Let $i:W\hookrightarrow \Lambda^m P$ be an affine subbundle and consider the affine subbundle $W_\alpha={\rm t}_\alpha\left(W\right)$. Let $\lambda_{W_\alpha}$ and $\lambda_W$ denote the restrictions of the
canonical $m$-form $\lambda_P$ to $W_\alpha$ and $W$ respectively. 
\begin{enumerate}[label={(\roman*)},leftmargin=*]
\item The following identity holds:
    \[
    \lambda_{W_\alpha}={\rm t}_{-\alpha}^*\lambda_W+i_\alpha^*\left(\pi_P^*\alpha\right)
    \]
where $i_\alpha\colon W_\alpha\hookrightarrow \Lambda^m P$ is the  inclusion.
\item Let $\Gamma_\alpha\colon M\rightarrow W_\alpha$ be a map and $X$ a vector field on $W_\alpha$ such that
  \[
  \Gamma_\alpha^*\left(X\lrcorner d\lambda_{W_\alpha}\right)=0.
  \]
Then for $\Gamma={\rm t}_{-\alpha}\circ\Gamma_\alpha$ the following identity holds:
  \[
  \Gamma^*\big(\left(T{\rm t}_{-\alpha}\circ X\right)\lrcorner\left(d\lambda_{W}+di^*\left(\pi_P^*\alpha\right)\right)\big)=0.
  \]
\end{enumerate}
\end{lemma}
\begin{proof}
Let us first remark that ${\rm t}_{\alpha}$ is a diffeomorphism, and therefore $W_\alpha$ is indeed an affine subbundle. We will prove $(i)$, since $(ii)$ follows easily. Let $p= \pi_P(\beta)$ where $\beta\in W$; then $\alpha_p+\beta\in W_\alpha$ and for $m$ tangent vectors $v_1,\dots,v_m$ in $T_{\alpha_p+\beta}W_\alpha$
\[
 \left.\lambda_{W_\alpha}\right|_{\alpha_p+\beta}(v_1,\dots,v_m)=
 \alpha_p\left(T\pi_P(v_1),\dots,T\pi_P(v_m)\right)+\beta \left(T\pi_P(v_1),\dots,T\pi_P(v_m)\right).
\]
On the other hand, using that $\pi_P\circ {\rm t}_{-\alpha}=\pi_P$, 
\begin{align*}
 {\rm t}_{-\alpha}^*(\left.\lambda_W\right|_\beta)\left(v_1,\dots,v_m\right)
 &=\beta\left(T\pi_P\circ T{\rm t}_{-\alpha} (v_1),\dots,T\pi_P\circ T{\rm t}_{-\alpha}(v_m)\right)\\
 &=\beta\left(T\pi_P(v_1),\dots,T\pi_P(v_m)\right).
\end{align*}
Comparing the previous identities, the claim follows. 
\end{proof}

\subsection{Routh reduction for Lagrangian field theories}

We are going to prove the main reduction theorem. Essentially, it states that the solutions of the invariant LFT $(\pi\colon E\rightarrow M,L\eta)$ project onto solutions of the reduced Lagrangian field theory (with force): 
\begin{equation}\label{eq:reducedLFT}
\big((\overline{\pi}\circ\overline{p})\colon E/G_\mu\times\mathop{\rm Lin}{\left(\overline{\pi}^*TM,\widetilde{\g}\right)} \to M,\mathcal{R}_\mu^{\rm red},\beta_\mu^{\rm red}\big), 
\end{equation}
where $\mathcal{R}_\mu^{\rm red}=q^*\overline{\mathcal{R}}_\mu$ and $\beta_\mu^{\rm red}=q^*\overline{\beta_\mu}$ are the Routhian density and the force in the reduced jet bundle $J^1(\overline{\pi}\circ\overline{p})$.

Before the proving this result, we need to make some observations about the reduced LFT~\eqref{eq:reducedLFT}. The proofs, which are straightforward in coordinates, are omitted.
\begin{enumerate}[label={\arabic*)},nolistsep]
 \item Any solution $\Gamma\colon M\to W_{q^*\overline{\mathcal{R}}_\mu}$ of the reduced LFT~\eqref{eq:reducedLFT} takes values in $W^0_{q^*\overline{\mathcal{R}}_\mu}$.
 \item If a section $\Gamma\colon M\to W^0_{q^*\overline{\mathcal{R}}_\mu}$ satisfies
 \begin{equation*}
\Gamma^*\big(X\lrcorner \big(d\lambda^0_{q^*\overline{\mathcal{R}}_\mu}+{\beta^0_\mu}\big)\big)=0,\quad \text{for all } X\in\mathfrak{X}^{V((\overline{\pi}\circ\overline{p})_1\circ\pi_{q^*\overline{\mathcal{R}}_\mu})}(W^0_{q^*\overline{\mathcal{R}}_\mu}).
 \end{equation*}
and additionally $\pi_{q^*\overline{\mathcal{R}}_\mu}\circ\Gamma\colon M\to J^1(\overline{\pi}\circ\overline{p})$ is holonomic then $\Gamma$, considered as a section $\Gamma\colon M\to W_{q^*\overline{\mathcal{R}}_\mu}$, is a solution.
\end{enumerate}

The next result states that, looking at solutions for the original (unreduced) and the reduced LFT at the level of $\Lambda^m_2(F_\omega)$, they coincide. 
\begin{proposition}\label{prop:correspondance} The set of solutions of $(W^\mu_{L\eta},\lambda^\mu_{L\eta})$ is in one-to-one correspondence with the set of solutions of $(W^0_{q^*\overline{\mathcal{R}}_\mu},\lambda^0_{q^*\overline{\mathcal{R}}_\mu},\beta^0_\mu)$.
\end{proposition}

\begin{proof} It is convienient to look at the commutative Diagram~\eqref{dia:aux2} below. Let $\Gamma\colon M\to W^\mu_{L\eta}$ be a solution of $(W^\mu_{L\eta},\lambda^\mu_{L\eta})$, we consider its lift $\widehat{\Gamma}\colon M\rightarrow\left({\rm pr}_1^\omega\right)^*W_{L\eta}^\mu$ to $\left({\rm pr}_1^\omega\right)^*W_{L\eta}^\mu$ given by Lemma~\ref{lem:liftingGamma}, which satisfies
\[
\widehat{\Gamma}^*\left(Z'\lrcorner d\lambda_{L\eta}'\right)=0, \quad \text{for all } Z'\in\mathfrak{X}^{V(\pi_1\circ\pi_{L\eta}\circ\Pi_{L\eta})}\big(\left({\rm pr}_1^\omega\right)^*W_{L\eta}^\mu\big).
\]
Apply now Lemma~\ref{lem:TranslationSpaceForms} to $P=\Lambda^m_2(F_\omega)$, $\alpha=\widehat{\omega}_\mu$, $W=\left(\text{pr}_2^\omega\right)^*W^0_{q^*\overline{\mathcal{R}}_\mu}$ and $W_\alpha=\left(\text{pr}_1^\omega\right)^*W_L^\mu$ (we are using Proposition~\ref{prop:WlWR_relation}), and we get the relation:
  \[
  \big({\rm t}_{-\widehat{\omega}_\mu}\circ\widehat{\Gamma}\big)^*\Big(\big(T{\rm t}_{-\widehat{\omega}_\mu}\circ Z'\big)\lrcorner d\lambda_{q^*\overline{\mathcal{R}}_\mu}'+d(i_{q^*\overline{\mathcal{R}}_\mu}^\omega)^*\psi^*\widehat{\omega}_\mu\Big)=0.
  \]
  
Note that ${\rm t}_{-\widehat{\omega}_\mu}\colon \Lambda^m_2(F_\omega)\to \Lambda^m_2(F_\omega)$ is a diffeomorphism which preserves the fibers of $\Lambda^m_2(F_\omega)\to F_\omega$ (and such that Diagram~\eqref{dia:aux2} commutes). Therefore 
\[
(T{\rm t}_{-\widehat{\omega}_\mu}\circ Z')\in \mathfrak{X}^{V((\overline{\pi}\circ\overline{p})_1\circ\pi_{q^*\overline{\mathcal{R}}_\mu})}\big(W^0_{q^*\overline{\mathcal{R}}_\mu}\big), 
\]
and any vector field in this set is of in this form. 
\begin{equation}\label{dia:aux2}
\begin{tikzcd}[column sep=1.5cm, row
    sep=1.5cm]
    & \Lambda^m_2 (F_\omega) & \\
    ({\rm pr}_1^\omega)^*W_{L\eta}^\mu \arrow[rr,swap,"{\rm t}_{-\widehat{\omega}_\mu}",bend right=15]\arrow[dr] \arrow[d,swap,"\Pi_{L\eta}"] \arrow[ur,hookrightarrow] & & (\text{pr}_2^\omega)^*W^0_{q^*\overline{\mathcal{R}}_\mu}\arrow[dl] \arrow[ll,swap,"{\rm t}_{\widehat{\omega}_\mu}",bend right=15] \arrow[d,"\Pi_{q^*\overline{\mathcal{R}}_\mu}"] \arrow[ul,hook'] \\
    W_{L\eta}^\mu \arrow[d,swap,"\pi_{L\eta}"] & F_\omega \arrow[dr,"{\rm pr}_2^\omega"]\arrow[dl,swap,"{\rm pr}_1^\omega"] & W^0_{q^*\overline{\mathcal{R}}_\mu} \arrow[d,"\pi_{q^*\overline{\mathcal{R}}_\mu}"]\\
    J^1\pi\arrow[dr,swap,"\pi_1"] &  & J^1(\overline{\pi}\circ\overline{p})\arrow[dl,"(\overline{\pi}\circ\overline{p})_1"]\\
    & M &
\end{tikzcd} 
\end{equation}
Applying Lemma \ref{lem:liftingGamma2} (see Remark~\ref{rem:1}), this means that
  \[
  \Pi_{q^*\overline{\mathcal{R}}_\mu}\circ {\rm t}_{-\widehat{\omega}_\mu}\circ\widehat{\Gamma}\colon M\to W^0_{q^*\overline{\mathcal{R}}_\mu}
  \]
is a solution of $(W^0_{q^*\overline{\mathcal{R}}_\mu},\lambda^0_{q^*\overline{\mathcal{R}}_\mu},\beta^0_\mu)$ , where we have used Lemma~\ref{lem:GyroForcesRelation} to identify
\[
d(i_{q^*\overline{\mathcal{R}}_\mu}^\omega)^*\psi^*\widehat{\omega}_\mu=\Pi_{q^*\overline{\mathcal{R}}_\mu}^*\beta^0_\mu.
\]
\end{proof}

We can now prove the main Routh reduction theorem for first order Lagrangian field theories.
  
\begin{theorem}[Reduction]\label{thm:main} Let $(\pi\colon E\rightarrow M,L\eta)$ be a $G$-invariant LFT and fix a (closed) value of the momentum map $\widehat{\mu}\in\Omega_1^{m-1}(W_{L\eta},\lag^*)$ and a principal connection $\omega$ on $E\to E/G$. Consider the reduced LFT
\[
\big((\overline{\pi}\circ\overline{p})\colon E/G_\mu\times\mathop{\rm Lin}{\left(\overline{\pi}^*TM,\widetilde{\g}\right)} \to M,\mathcal{R}_\mu^{\rm red},\beta_\mu^{\rm red}\big).
\]
Then every solution of the LFT $(\pi\colon E\rightarrow M,L\eta)$ with momentum $\widehat{\mu}$ projects onto a solution of the reduced LFT. The reduced solution is given by $\gamma^{\rm red}=f_\omega\circ\gamma$.
\end{theorem}

\begin{proof} If $\gamma\colon M\to E$ is a solution of the LFT $(\pi\colon E\rightarrow M,L\eta)$, then we construct $\Gamma\colon M\to W_{L\eta}$ solution of $(W_{L\eta},\lambda_{L\eta})$. By the momentum constraint, we have $\Gamma\colon M\to W^\mu_{L\eta}$, so $\Gamma$ is a solution of $(W^\mu_{L\eta},\lambda^\mu_{L\eta})$. Now we apply Proposition~\ref{prop:correspondance} and we get a solution $\Gamma^{\rm red}$ of $(W^0_{q^*\overline{\mathcal{R}}_\mu},\lambda^0_{q^*\overline{\mathcal{R}}_\mu},\beta^0_\mu)$. 

By diagram chasing (see the proof of Lemma~\ref{lem:liftingGamma}), it is not hard to see that 
\[
\pi_{q^*\overline{\mathcal{R}}_\mu}\circ\Gamma^{\rm red}=j^1\gamma^{\rm red}, 
\]
with $\gamma^{\rm red}=f_\omega(\gamma)$. In view of the observations above, this means that $\widehat{\Gamma}$ is a solution of the variational problem on $W_{q^*\overline{\mathcal{R}}_\mu}$, and therefore $\gamma^{\rm red}\colon M\to E/G_\mu\times\mathop{\rm Lin}{\left(\overline{\pi}^*TM,\widetilde{\g}\right)}$ is a solution of the reduced LFT.
\end{proof}

\section{Reconstruction}\label{sec:reconstruction}

In general, the problem of reconstruction in geometric reduction addresses the following two questions: 
\begin{enumerate}[label={\arabic*)},nolistsep]
\item  Given a solution of the reduced system, is it always possible to find a solution of the original (unreduced) system projecting onto it?
\item If the answer to the previous question is affirmative, how does one effectively construct such a solution?
\end{enumerate}

In Lagrangian mechanics, both the Lagrange-Poincar\'e and the Routh reduction schemes provide reduced systems which are equivalent to the unreduced ones\footnote{In the case of Routh reduction, one reconstructs \emph{only} solutions with a fixed momentum.}. However, in the case of Lagrangian field theory, this is not the case as was first observed in~\cite{Castrillon0} in the context of Euler-Poincar\'e reduction (i.e. Lagrange-Poincar\'e reduction for a Lie group). In this section, we will show that in the case of Routh reduction for field theories there is also an obstruction to reconstruction which coincides with that of the Lagrange-Poincar\'e case~\cite{Castrillon1,Castrillon2,LagPoincare_JGP}.

\subsection{Lifting sections on reduced jet bundles}

Consider a $G$-principal fiber bundle $p^P_G\colon P\to P/G$ for which there are two fibrations $\pi\colon P\to M$ and $\overline{\pi}\colon P/G\to M$ making Diagram~\eqref{dia:covering} (left) commutative. Given a section $\zeta\colon M\to P/G$, we want to find conditions to ensure that there exists a section $s\colon M\to P$ covering $\zeta$, i.e. such that $p^P_G\circ s=\zeta$. To answer this question, we look at the pullback bundle $\zeta^*P$, see  Diagram~\eqref{dia:covering} (right). 
\begin{equation}\label{dia:covering}
 \begin{tikzcd}[column sep=1cm, row
    sep=1.4cm]
    P \arrow[rr,"p^P_G"]\arrow[dr,swap,"\pi"] & & [1ex]P/G \arrow[dl,"\overline{\pi}"]\\
    & M \arrow[ur,dashed,bend right=45,swap,"\zeta"]\arrow[ul,dashed,bend left=45,"s"]&
  \end{tikzcd}\hspace{3em}
  \begin{tikzcd}[column sep=1.8cm, row
    sep=1.4cm]
    \zeta^*P \arrow[r,"\text{pr}_2"]\arrow[d,swap,"\text{pr}_1"] &[1ex] P\arrow[d,"p^P_G"]\\
     M\arrow[r,swap,"\zeta"] & P/G 
  \end{tikzcd}
\end{equation}

The pullback bundle $\text{pr}_1\colon \zeta^*P\to M$ is a $G$-principal bundle with action 
\[
 g\cdot (x,p)=(x,g\cdot p).
\]
For later use, we recall that tangent space $T_{\left(x,p\right)}\zeta^*P$ at $\left(x,p\right)\in\zeta^*P$ is given by
\[
T_{\left(x,p\right)}\zeta^*P=\left\{\left(v_x,V_p\right):T_x\zeta\left(v_x\right)=T_pp_G^P\left(V_p\right)\right\}\subset T_xM\times T_pP.
\]
\begin{lemma}\label{lem:TrivialToSections}
There exists a section $s:M\rightarrow P$ covering the section $\zeta\colon M\to P/G$ if and only if $\zeta^*P$ is a trivial bundle.
\end{lemma}
\begin{proof} Since $\zeta^*P$ is principal, it is trivial if and only if it admits a section $\tilde s\colon M\to \zeta^*P$. If $\tilde s$ exists, then $s=\text{pr}_2\circ \tilde s$ is the desired section. Conversely, if $s\colon M\to P$ exists, then $\tilde s\colon M\to \zeta^*P\subset M\times P$ is defined by $\tilde s(x)=(x,s(x))$.
\end{proof}

Using that $\zeta^*P$ is a principal bundle,  being trivial can be characterized in terms of a flat connection~\cite{KN}:

\begin{theorem}\label{thm:conn-induc-pullb}
  Let $\pi:P\rightarrow M$ be a $G$-principal bundle with $M$ simply connected. Then $P$ is trivial if and only if there exists a flat connection on $P$. 
\end{theorem}

\begin{proof} Obviously if $P\simeq M\times H$ is trivial we can consider the canonical flat connection on $P$. Conversely, given a flat connection we take an integral leaf $L$ of the horizontal distribution and $\pi^{-1}(x)\cap L$ has a unique element (since $M$ is simply connected, the connection has trivial honolomy), and this defines a section of $P\to M$. 
\end{proof}
If $M$ is not simply connected, then one can ask for a flat connection with trivial holonomy and obtain a similar result. For the sake of simplicity, we will assume that $M$ is simply connected to apply Theorem~\ref{thm:conn-induc-pullb} when needed. For later use, we also observe that the section constructed in the proof of Theorem~\ref{thm:conn-induc-pullb} has horizontal image w.r.t. the given connection.

We now wish to apply the previous discussion to the case of jet bundles. We start with the first jet $J^1\pi$ of a bundle $\pi\colon P\to M$ and construct the quotients $P/G$ and $J^1\pi/G$.  More concretely, we look at the situation is depicted in Diagram~\eqref{dia:covering2} (left): $Z\colon M\to J^1\pi/G$ is a given section  and $\zeta\colon M\to P/G$ is the induced section. The basic question we want to address is the following: does there exist a holonomic section $\widehat{Z}:M\rightarrow J^1\pi$ such that $p_G^{J^1\pi}\circ\widehat{Z}=Z$?
\begin{equation}\label{dia:covering2}
  \begin{tikzcd}[column sep=1cm, row
    sep=1cm]
    J^1\pi \arrow[rr,"p^{J^1\pi}_G"]\arrow[d,"\pi_{10}"] &&J^1\pi/G\arrow[d,swap,"\overline{\pi}_{10}"]\\
     M\arrow[rr,"p^P_G"]\arrow[dr,"\pi"] & & J^1\pi/G \arrow[dl,swap,"\overline{\pi}"] \\
     & M\arrow[ur,swap,dashed,"\zeta",bend right=25]\arrow[uul,dashed,"\widehat{Z}",bend left=85]\arrow[uur,swap,dashed,"Z",bend right=85]&
  \end{tikzcd}
  \hspace{1em}
  \begin{tikzcd}[column sep=2cm, row
    sep=1.8cm]
    Z^*\left(J^1\pi\right) \arrow[r,"\text{pr}_2"]\arrow[d,swap,"\text{pr}_1"] &J^1\pi\arrow[d,"p_G^{J^1\pi}"]\\
     M\arrow[r,swap,"Z"] & J^1\pi/G
  \end{tikzcd}
\end{equation}
We remark that $J^1\pi\to J^1\pi/G$ is a principal bundle. We can then construct the pullback bundle $Z^*(J^1\pi)$ (Diagram~\eqref{dia:covering2}, right) and particularize Lemma~\ref{lem:TrivialToSections} to conclude the following:
\begin{lemma}\label{lem:FlatConnectionIffSectionJ1pi}
Assume that $M$ is simply connected. Then $Z^*\left(J^1\pi\right)$ admits a flat connection if and only if there exists a section $\widehat{Z}:M\rightarrow J^1\pi$.
\end{lemma}

There is also a map $r:J^1\pi/G\rightarrow J^1\overline{\pi}$, $[j^1_x s]_G\mapsto j^1_x[s]_G$ making the following diagram commutative
\begin{equation}\label{eq:QpJ1piJ1opi}
    \begin{tikzcd}[ampersand replacement=\&, column sep=1.8cm, row sep=1.5cm]
      J^1\pi\arrow{r}{p_G^{J^1\pi}}
      \arrow{dr}[swap]{j^1p_G^P}
      \&
      J^1\pi/G
      \arrow{d}{r}
      \\
      \&
      J^1\overline{\pi}
    \end{tikzcd}
\end{equation}
As before, we denote by $\theta\in\Omega^1(J^1\pi,V\pi)$ the canonical contact form on $J^1\pi$.

\begin{theorem}\label{thm:conn-lift-sect}
Let $Z:M\rightarrow J^1\pi/G$ be a section of the quotient bundle such that
\[
r\left({Z}\left(x\right)\right)=j^1_x\zeta
\]
  where $\zeta:=\overline{\pi}_{10}\circ Z:M\rightarrow P/G$.
  \begin{enumerate}[label={\arabic*)},nolistsep]
  \item Suppose that there exists a holonomic section $\widehat{Z}:M\rightarrow J^1\pi$ such that $p_G^{J^1\pi}\circ\widehat{Z}=Z$. Then for any connection $\omega_P$ on the principal bundle $p_G^P:P\rightarrow P/G$, the connection
    \[
    \omega^Z=\omega_P\circ\left({\rm pr}_2\right)^*\theta\in\Omega^1\left(Z^*(J^1\pi),\lag\right)
    \]
    is a flat connection on $Z^*\left(J^1\pi\right)$.
  \item Conversely, suppose that for some connection $\omega_P$ on $P\to P/G$ (and hence for all) the connection
    \[
    \omega^Z=\omega_P\circ\left({\rm pr}_2\right)^*\theta\in\Omega^1\left(Z^*(J^1\pi),\lag\right)
    \]
    is a flat connection on $Z^*\left(J^1\pi\right)$. Then the associated section $\widehat{Z}:M\rightarrow J^1\pi$ through Lemma \ref{lem:FlatConnectionIffSectionJ1pi} is holonomic.
  \end{enumerate}
\end{theorem}

\begin{proof} If $\phi_g\colon P\to P$ denotes the action $g\cdot p=\phi_g(p)$, then $L_g\colon J^1\pi\to J^1\pi$ denotes the prolonged action $g\cdot j^1_xs=L_g(j^1_xs)=j^1\phi_g(j^1_xs)$.

\noindent 1) Using the section $\widehat{Z}:M\rightarrow J^1\pi$, we have a  connection on $Z^*\left(J^1\pi\right)$ whose horizontal subspaces at $(x,j^1_x s=g\cdot \widehat{Z}(x))$ are given by:
    \[
    H^{\widehat{Z}}_{\left(x,j_x^1s\right)}=\left\{\big(v_x,T_{\widehat{Z}(x)}L_g\circ T_x\widehat{Z}(v_x)\big):v_x\in T_xM\right\},
    \]
This distribution is integrable: this follows from the fact that brackets of left-invariant vector fields are left-invariant. Hence $H^{\widehat{Z}}$ is a flat connection. Since $\widehat{Z}$ is holonomic, it satisfies $\widehat{Z}^*\theta=0$ and then using the invariance of the contact structure, $(L_g)^*\theta=\theta$, we find:
    \begin{align*}
    \left.\left(\text{pr}_2\right)^*\theta\right|_{\left(x,j_x^1\right)}\big(v_x,T_{\widehat{Z}(x)}L_g\circ T_x\widehat{Z}\left(v_x\right)\big)&=\left.\theta\right|_{j_x^1s}\big(T_{\widehat{Z}(x)}L_g\circ T_x\widehat{Z}\left(v_x\right)\big)\\&=(\widehat{Z}^*\theta)|_x(v_x)=0.
    \end{align*}
This implies that 
\begin{equation}\label{eq:eqaux1}
H^{\widehat{Z}}\subset\ker(\omega^{Z}).   
\end{equation}
Next, we will show that  $\dim \ker(\omega^{Z})\leq \dim M$. Together with~\eqref{eq:eqaux1}, this proves that $ H^{\widehat{Z}}=\ker(\omega^{Z})$.

We observe the following: $\theta$ is $V\pi$-valued and $\omega_P$ restricts to the identity on $V\pi$. Therefore a tangent vector $(v_x,V_{j^1_xs})$ to $Z^*(J^1\pi)$ at a point $(x,j^1_xs)$ $Z^*(J^1\pi)$ will belong to $\ker(\omega^Z)$ if, and only if, $0=\text{pr}_2^*\theta(v_x,V_{j^1_xs})=\theta(V_{j^1_xs})$. Let us assume that we have two different tangent vectors $(v_x,V_{j^1_xs})$ and $(v_x,W_{j^1_xs})$ such that $\theta(V_{j^1_xs})=\theta(W_{j^1_xs})=0$. Then the conditions $T_xZ(v_x)=T_{j^1_xs}p^{J^1\pi}_G(V_{j^1_xs})$ and $T_xZ(v_x)=T_{j^1_xs}p^{J^1\pi}_G(W_{j^1_xs})=0$ imply that $(V_{j^1_xs}-W_{j^1_xs})$ is vertical w.r.t. $p^{J^1\pi}_G$. But then it is of the form $\xi_{J^1\pi}(j^1_xs)$ for some $\xi\in\lag$ different from $0$. Since $\xi_{J^1\pi}=j^1\xi_P$ is the prolongation of the vertical vector field $\xi_P$, it follows that $\theta(\xi_{J^1\pi})=\xi_P\neq 0$ (this can be checked easily with the coordinate expression of the prolongation, see~\cite{GeometryLagrangianFieldTheories}), which is not possible. Hence for each $v_x\in T_xM$ there is at most one choice of $V_{j^1_xs}$ such that $\theta(V_{j^1_xs})=0$ and $(v_x,V_{j^1_xs})\in T_{(x,j^1_xs)}Z^*(J^1\pi)$. It follows that $\dim \ker(\omega^{Z})\leq \dim M$.

\noindent 2) If $\omega^{Z}$ is a flat connection on $Z^*(J^1\pi)$, let $\widehat{Z}:M\rightarrow J^1\pi$ be the corresponding section via Lemma \ref{lem:FlatConnectionIffSectionJ1pi}. We have already shown that
  \[
    \ker(\omega^Z)_{\left(x,j_x^1s\right)}=\left\{\big(v_x,T_{\widehat{Z}(x)}L_g\circ T_x\widehat{Z}(v_x)\big):v_x\in T_xM\right\},
\]
where $(x,j^1_x s=g\cdot \widehat{Z}(x))$. Recalling the definition of the contact structure~\eqref{eq:contactstructuredefinition}, we have 
    \begin{align*}
      \left.\left(\text{pr}_2\right)^*\theta\right|_{\left(x,j^1_xs\left(x\right)\right)}\left(v_x,T_{\widehat{Z}\left(x\right)}L_g\circ T_x\widehat{Z}\left(v_x\right)\right)&=\left.\theta\right|_{j^1_xs}\left(T_{\widehat{Z}\left(x\right)}L_g\circ T_x\widehat{Z}\left(v_x\right)\right)\\
      &=T_{\widehat{Z}\left(x\right)}\pi_{10}\left(T_x\widehat{Z}\left(v_x\right)\right)-\widehat{Z}\left(x\right)\left(v_x\right),
    \end{align*}
where in the last term we have interpreted the element $\widehat{Z}\left(x\right)\in J^1\pi$ as a map
    \[
    \widehat{Z}\left(x\right):T_xM\rightarrow T_{\pi_{10}\left(\widehat{Z}\left(x\right)\right)}P.
    \]
Thus the condition $\big(v_x,T_{\widehat{Z}(x)}L_g\circ T_x\widehat{Z}(v_x)\big)\in\ker(\omega^Z)$ reads
\[
\omega_P\big(T_{\widehat{Z}((x)}\pi_{10}\big(T_x\widehat{Z}\left(v_x\right)\big)-\widehat{Z}\left(x\right)\left(v_x\right)\big)=0.
\]
The section $s=\pi_{10}\circ\widehat{Z}:M\rightarrow P$ satisfies $p_G^P\circ s=\zeta$, and the condition above means that there exists $\Gamma\colon TM\to \ker\omega_P\subset TP$ with  
    \[
    \widehat{Z}\left(x\right)=T_xs+\left.\Gamma\right|_x
    \]
Projecting along the map $Tp_G^P$, we have that
    \[
    T_{s\left(x\right)}p_G^P\circ\widehat{Z}\left(x\right)=T_{s\left(x\right)}p_G^P\circ T_xs+T_{s\left(x\right)}p_G^P\left(\left.\Gamma\right|_x\right).
    \]
From Diagram \eqref{eq:QpJ1piJ1opi}, we have that
    \[
    T_{s\left(x\right)}p_G^P\circ\widehat{Z}\left(x\right)=q\circ p_G^{J^1\pi}\\
    \big(\widehat{Z}\left(x\right)\big)=q\left(Z\left(x\right)\right)=T_x\zeta,
    \]
and also, since $r\circ Z=T\zeta$, 
    \[
    T_{s\left(x\right)}p_G^P\circ T_xs=T_x\left(p_G^P\circ s\right)=T_x\zeta.
    \]
It follows that $T_{s\left(x\right)}p_G^P\left(\left.\Gamma\right|_x\right)=0$. Because $Tp_G^P$ is an isomorphism when restricted to $\ker{\omega_P}$ (recall that $\omega_P$ is a connection for the bundle $p_G^P:P\rightarrow P/G$), it means that $\left.\Gamma\right|_x=0$, and so
    \[
    \widehat{Z}\left(x\right)=T_xs,
    \]
i.e., $\widehat{Z}$ is a holonomic section.
\end{proof}

\begin{remark} The fact that the section $\widehat{Z}$ determines the horizontal distribution of the connection $\omega^Z$ is referred to as the \emph{horizontality condition} in~\cite{LagPoincare_JGP}. 
\end{remark}

\subsection{The case of Routh reduction}

We will now apply the previous constructions and results about liftings on sections to find conditions for reconstruction. To reconstruct, one should reverse the proof of Theorem~\ref{thm:main}. The key point if to find the analog of Lemma~\ref{lem:liftingGamma} for reduced sections, which requires additional conditions.

Consider a section $\gamma^{\rm red}\colon M\to E/G_\mu\times\mathop{\rm Lin}{\left(\overline{\pi}^*TM,\widetilde{\g}\right)}$ of the reduced LFT. It gives rise to a section $\overline{Z}\colon M\to J^1\overline{\pi}\times E/G_\mu\times\mathop{\rm Lin}{\left(\overline{\pi}^*TM,\widetilde{\g}\right)}$ obtained as $\overline{Z}=q(j^1\gamma^{\rm red})$ (see Diagram~\eqref{dia:Pullback_bundle}). In the following definition we particularize the integrability condition on Theorem~\ref{thm:conn-lift-sect} for lifting the section $Z$ to $J^1\pi$. Recall that there is an identification
\[
g_\omega\colon J^1\pi/G_\mu\to  J^1\overline{\pi}\times E/G_\mu\times\mathop{\rm Lin}{\left(\overline{\pi}^*TM,\widetilde{\g}\right)},
\]
and thus we have a section $Z\colon M\to J^1\pi/G_\mu$ given by $Z=g_\omega^{-1}(\overline{Z})$. As before, $\theta$ denotes the canonical contact form on $J^1\pi$ and ${\rm pr}_2\colon Z^*(J^1\pi)\to J^1\pi$ is the canonical projection.

\begin{definition} Let $\gamma^{\rm red}\colon M\to E/G_\mu\times\mathop{\rm Lin}{\left(\overline{\pi}^*TM,\widetilde{\g}\right)}$ be a section of the reduced LFT and $Z=g_\omega^{-1}( q(j^1\gamma^{\rm red}))$. Fix a principal connection on $\omega$ on $E\to E/G$. We will say that $\gamma^{\rm red}$ satisfies the \emph{flat condition} if the connection
\[
\omega^Z=\omega\circ{\rm pr}_2^*\theta \in \Omega^1\big(Z^*(J^1\pi),\lag\big) 
\]
is flat. 
\end{definition}

We will now check that if  $\gamma^{\rm red}\colon M\to E/G_\mu\times\mathop{\rm Lin}{\left(\overline{\pi}^*TM,\widetilde{\g}\right)}$ satisfies the flat condition, then the associated section $\Gamma\colon M\to W^0_{q^*\overline{\mathcal{R}}_\mu}$ can be lifted to $({\rm pr}_2)^*W^0_{q^*\overline{\mathcal{R}}_\mu}$. Indeed, to mimic the proof of Lemma~\ref{lem:liftingGamma} all one needs to do is to find a lift $\widetilde Z\colon M\to F_\omega$ of $\gamma^{\rm red}$. To find such a lift, it suffices to find a holonomic lift $\widehat{Z}\colon M\to J^1\pi$ of the section $Z= g_\omega^{-1}( q(j^1\gamma^{\rm red}))$. The situation is summarized in the following diagram:
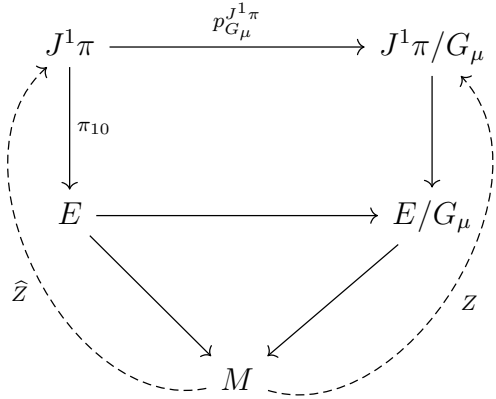
\begin{equation*}
  \begin{tikzcd}[column sep=1.3cm, row sep=1.5cm]
   J^1\pi\arrow[d,"\pi_{10}"]\arrow[rr,"p^{J^1\pi}_{G_\mu}"] & &  J^1\pi/G_\mu\arrow[d,swap,""]\\
  E\arrow[dr,swap]\arrow[rr] & &  E/G_\mu \arrow[dl,""]\\
  & M\arrow[uur,dashed,swap,"Z",bend right=75]\arrow[uul,dashed,"\widehat{Z}",bend left=75] &  
  \end{tikzcd}
\end{equation*}
Using Theorem~\ref{thm:conn-lift-sect}, such a lift $\widehat{Z}$ exists if and only if $\gamma^{\rm red}$ satisfies the flat condition. To conclude, given $\gamma^{\rm red}\colon M\to E/G_\mu\times\mathop{\rm Lin}{\left(\overline{\pi}^*TM,\widetilde{\g}\right)}$ which satisfies the flat condition we construct $\widehat{Z}\colon M\to J^1\pi$ and the section
\begin{align*}
\widetilde Z\colon M&\to F_\omega,\\
x&\mapsto (\widehat{Z}(x),j^1\gamma^{\rm red}(x)),
\end{align*}
is the desired section. Therefore, one can prove the following:
\begin{lemma}\label{lem:liftingGamma3} Let $\gamma^{\rm red}\colon M\to E/G_\mu\times\mathop{\rm Lin}{\left(\overline{\pi}^*TM,\widetilde{\g}\right)}$ be a section of the reduced LFT which satisfies the flat condition and let $\Gamma\colon M\to W^0_{q^*\overline{\mathcal{R}}_\mu}$ be the associated section. Then there exists a section
\[
\widehat{\Gamma}\colon M\rightarrow\left({\rm pr}_2^\omega\right)^*W^0_{q^*\overline{\mathcal{R}}_\mu} 
\]
such that $\Gamma=\Pi_{q^*\overline{\mathcal{R}}_\mu}\circ \widehat{\Gamma}\colon M\to W^0_{q^*\overline{\mathcal{R}}_\mu}$ and which satisfies
\[
\widehat{\Gamma}^*\left(Z'\lrcorner d\lambda_{L\eta}'\right)=0, \quad \text{for all } Z'\in\mathfrak{X}^{V((\overline{\pi}\circ\overline{p})_1\circ\pi_{q^*\overline{\mathcal{R}}_\mu}\circ\Pi_{q^*\overline{\mathcal{R}}_\mu})}\big(({\rm pr}_2^\omega)^*W^0_{q^*\overline{\mathcal{R}}_\mu}\big).
\]
\end{lemma}
\begin{proof} One can mimic the proof of Lemma~\ref{lem:liftingGamma} using the lift to $F_\omega$ above.
\end{proof}

\begin{theorem}[Reconstruction]\label{thm:main2} Let $\gamma^{\rm red}\colon M\to E/G_\mu\times\mathop{\rm Lin}{\left(\overline{\pi}^*TM,\widetilde{\g}\right)}$ be a solution of the reduced LFT  
\[
\big((\overline{\pi}\circ\overline{p})\colon E/G_\mu\times\mathop{\rm Lin}{\left(\overline{\pi}^*TM,\widetilde{\g}\right)} \to M,\mathcal{R}_\mu^{\rm red},\beta_\mu^{\rm red}\big).
\]
which satisfies the flat condition. Then there exist a solution of the unreduced LFT $(\pi\colon E\to M,L\eta)$ with momentum $\widehat{\mu}$ which projects onto it. 
\end{theorem}

\begin{proof} Given $\gamma^{\rm red}\colon M\to E/G_\mu\times\mathop{\rm Lin}{\left(\overline{\pi}^*TM,\widetilde{\g}\right)}$, one uses Lemma~\ref{lem:liftingGamma3} to construct
\[
\widehat{\Gamma}^{\rm red}\colon M\rightarrow\left({\rm pr}_2^\omega\right)^*W^0_{q^*\overline{\mathcal{R}}_\mu}  
\]
and reverses the proof of Theorem~\ref{thm:main} to find a solution $\Gamma\colon M\to ({\rm pr}_1)^*W^\mu_{L\eta}$. Since by construction $\pi_{L\eta}\circ \Gamma=\widehat Z$ is holonomic, $\Gamma$ is also a solution when regarded as a section      $\Gamma\colon M\to ({\rm pr}_1)^*W_{L\eta}$. Thus, letting $\widehat Z=j^1\gamma$, it follows that  $\gamma\colon M\to E$ is a solution of the unreduced LFT.
\end{proof}

\section{An example related to the KdV equation}\label{sec:Example}

We will revisit the last example in~\cite{Castrillon2} from the perspective of Routh reduction. Consider the bundle $\pi\colon E=\R^2\times\R^2\to M=\R^2$ with coordinates $(t,x,\phi,\psi)$ and the Lagrangian on $J^1\pi$ given by
\[
L\left(t,x,\phi,\psi,\phi_t,\phi_x,\psi_t,\psi_x\right)=\frac{1}{2}\phi_t\phi_x+\phi_x^3+\phi_x\psi_x+\frac{1}{2}\psi^2.
\]
We take the volume form $\eta=dt\wedge dx$ on $M$. The Lagrangian is invariant by the action of the Lie group $G=\R$ by translations on $\phi$. The infinitesimal generator of $\xi\in\lag$ is $\xi_{W_{L\eta}}=\xi\partial_{\phi}$. The bundle $W_{L\eta}$ has coordinates $(t,x,\phi,\psi,\phi_t,\phi_x,\psi_t,\psi_x,p_\phi^t,p_\phi^x,p_\psi^t,p_\psi^x)$ so that the canonical form reads
\begin{multline*}
  \lambda_{L\eta}=\left(\frac{1}{2}\phi_t\phi_x+\phi_x^3+\phi_x\psi_x+\frac{1}{2}\psi^2\right)dt\wedge dx+p^t_\phi\left(d\phi-\phi_tdt\right)\wedge dx -\\
  -p^x_\phi\left(d\phi-\phi_xdx\right)\wedge dt+p^t_\psi\left(d\psi-\psi_tdt\right)\wedge dx-p^x_\psi\left(d\psi-\psi_xdx\right)\wedge dt,
\end{multline*}
and then the momentum map is easily found to be
\[
J(t,x,\phi,\psi,\phi_t,\phi_x,\psi_t,\psi_x,p^t_\phi,p^x_\phi,p^t_\psi,p^x_\psi)=p^t_\phi dx-p^x_\phi dt.
\]
We now fix a momentum value $\mu=\mu_1dt+\mu_2dx$ which should be closed, i.e. $d\mu=0$, which implies $\partial \mu_1/\partial x=\partial \mu_2/\partial t$. Clearly, the submanifold $W_{L\eta}^\mu\subset W_{L\eta}$ is described by $\{p^t_\phi=\mu_2, p^x_\phi=-\mu_1\}$. The isotropy group is $G_\mu=G=\R$.

To construct the Routhian, we choose a principal connection on $E\to E/G\simeq \R^3$ which will be of the form
\[
\omega=d\phi-\Gamma_t(t,x,\psi)dt-\Gamma_x(t,x,\psi)dx-\Gamma_\psi(t,x,\psi)d\psi.
\]

We have:
\[
 \omega\circ T_{\left(t,x\right)}s=(\phi_t-\Gamma_t-\Gamma_\psi\psi_t)dt+(\phi_x-\Gamma_x-\Gamma_\psi\psi_x)dx.
\]
The (unreduced) Routhian is then (note that $\varepsilon=-1$)
\begin{align*}
R_\mu(&t,x,\phi,\psi,\phi_t,\phi_x,\psi_t,\psi_x)=L\eta-(-1)\langle\mu\stackrel{\wedge}{,}\omega\circ T_{\left(t,x\right)}s\rangle\\
&=L\eta+\left(\mu_1 dt+\mu_2 dx\right)\wedge\left[\left(\phi_t-\Gamma_t-\Gamma_\psi\psi_t\right)dt+\left(\phi_x-\Gamma_x-\Gamma_\psi\psi_x\right)dx\right]\\
&=\frac{1}{2}\phi_t\phi_x+\phi_x^3+\phi_x\psi_x+\frac{1}{2}\psi^2+\mu_1\left(\phi_x-\Gamma_x-\Gamma_\psi\psi_x\right)-\mu_2\left(\phi_t-\Gamma_t-\Gamma_\psi\psi_t\right).
\end{align*}

The fibration $\overline{\pi}\colon E/G_\mu=\R^2\times\R\to M=\R^2$ is
\[
\overline{\pi}(t,x,\psi)=(t,x).
\]
Hence
\[
E/G_\mu\times\mathop{\text{Lin}}{\left(\overline{\pi}^*TM,\widetilde\g\right)}=(\R^2\times\R)\times_{\R^2}T^*\R^2=T^*\R^2\times\R,
\]
and also $\overline{p}\colon E/G_\mu\times\mathop{\text{Lin}}{\left(\overline{\pi}^*TM,\widetilde\g\right)}=T^*\R^2\times\R\to E/G=\R^2\times\R$ is
\[
\overline{p}(t,x,\sigma,\rho,\psi)=(t,x,\psi),
\]
where $(t,x,\sigma,\rho)$ are coordinates on $T^*\R^2$. The projection $(\overline{\pi}\circ \overline{p})$ is simply
\[
(\overline{\pi}\circ \overline{p})(t,x,\sigma,\rho,\psi)=(t,x)
\]
and $q:J^1\left(\overline{\pi}\circ\overline{p}\right)\to J^1\overline{\pi}\times E/G_\mu\times\mathop{\text{Lin}}{\left(\overline{\pi}^*TM,\widetilde\g\right)}=J^1\overline{\pi}\times T^*\R^2$ is
\[
q(t,x,\sigma,\rho,\psi,\sigma_t,\sigma_x,\rho_t,\rho_x,\psi_t,\psi_x)=(t,x,\psi,\psi_t,\psi_x,\sigma,\rho).
\]

On the other hand, $J^1\pi/G_\mu$ with coordinates $(t,x,\psi,\phi_t,\phi_x,\psi_t,\psi_x)$, and looking at how 
\[
g_\omega\colon J^1\pi/{G_\mu} \to J^1\overline{\pi}\times E/{G_\mu}\times{\rm Lin}{\left(\overline{\pi}^*TM,\widetilde{\g}\right)}
\]
is defined (Section~\ref{sec:RouthLFT}), we find
\[
g_\omega\left(t,x,\psi,\phi_t,\phi_x,\psi_t,\psi_x\right)=\left(t,x,\psi,\psi_t,\psi_x,\sigma=\phi_t-\Gamma_t-\Gamma_\psi\psi_t,\rho=\phi_x-\Gamma_x-\Gamma_\psi\psi_x\right).
\]
Using this expression, $R_\mu^{\rm red}(t,x,\psi,\psi_t,\psi_x,\sigma,\rho)$ is easily obtained:
\begin{align*}
R_\mu^{\rm red}=\frac{1}{2}(\sigma+&\Gamma_t+\Gamma_\psi\psi_t)(\rho+\Gamma_x+\Gamma_\psi\psi_x)+\left(\rho+\Gamma_x+\Gamma_\psi\psi_x\right)^3 +\\
&+\left(\rho+\Gamma_x+\Gamma_\psi\psi_x\right)\psi_x+\frac{1}{2}\psi^2+\mu_1\rho-\mu_2\sigma.
\end{align*}

To compute the force term, we have  
\begin{align*}
\omega_\mu&=-\langle\mu\stackrel{\wedge}{,}\omega_E\circ T\pi_{10}\rangle=-\left(\mu_1 dt+\mu_2 dx\right)\wedge\left(d\phi-\Gamma_t dt-\Gamma_x dx-\Gamma_\psi d\psi\right)
\end{align*}
and therefore (since $\mu$ is closed)
\[
d\omega_\mu=\left(\mu_2\fpd{\Gamma_\psi}{t}-\mu_1\fpd{\Gamma_\psi}{x}\right)dt\wedge dx\wedge d\psi.
\]
The force $\beta^{\rm red}$ is obtained by pullback, and has the same coordinate expression as $d\omega_\mu$. Therefore, the Euler-Lagrange equations for $R_\mu^{\rm red}$ with force $\beta^{\rm red}$ are:
\begin{align*}
\frac{1}{2}\fpd{}{t}[\Gamma_\psi(\rho+\Gamma_x+\Gamma_\psi\psi_x)]+\frac{1}{2}\fpd{}{x}[\Gamma_\psi(\sigma+\Gamma_t+\Gamma_\psi\psi_t)]+3&\fpd{}{x}([\rho+\Gamma_x
+\Gamma_\psi\psi_x]\Gamma_\psi)+\\
+\fpd{}{x}[2\Gamma_\psi\psi_x+\rho+\Gamma_x]-\psi&=\left(\mu_2\fpd{\Gamma_\psi}{t}-\mu_1\fpd{\Gamma_\psi}{x}\right),\\
(\rho+\Gamma_x+\Gamma_\psi\psi_x)&=2\mu_2,\\
(\sigma+\Gamma_t+\Gamma_\psi\psi_t)+6\left(\rho+\Gamma_x+\Gamma_\psi\psi_x\right)^2+2\psi_x&=-2\mu_1.
\end{align*}

If we choose the canonical flat connection $\Gamma_t=\Gamma_x=\Gamma_\psi=0$ the Routhian becomes
\[
\widehat{R}_\mu^{\rm red}=\frac{1}{2}\sigma\rho+\rho^3+\rho\psi_x+\frac{1}{2}\psi^2+\mu_1\rho-\mu_2\sigma,
\]
and the Euler-Lagrange equations are
\[
\fpd{\rho}{x}=\psi,\quad \rho=2\mu_2,\quad \sigma+6\rho^2+2\psi_x=-2\mu_1.
\]
Differentiating the last and replacing $\psi$ using the first equation, we find
\[
\fpd{\rho}{t}=2\fpd{\mu_2}{t},\quad
\fpd{\sigma}{x}+12\rho\fpd{\rho}{x}+2\fpd{^3\rho}{x^3}=-2\fpd{\mu_1}{x}. 
\]
Now using that $\mu$ is closed we have $\partial \rho/\partial t=-2\partial \mu_1/\partial x$. If one further imposes the integrability condition $\partial\sigma/\partial x=\partial\rho/\partial t$ (which is needed in view of the definition of $g_\omega$), one finds that $\rho$ must satisfy the KdV equation:
\[
\fpd{\rho}{t}+6\rho\fpd{\rho}{x}+\fpd{^3\rho}{x^3}=0. 
\]
Note that this imposes that the chosen momentum $\mu_2$ must satisfy a PDE which, after scaling, is again of KdV type. The fact that this result can be directly compared to that of~\cite{Castrillon2} reflects the well-known result that, in the case of an Abelian Lie group of symmetries, there is a close relation between the Lagrange-Poincar\'e and the Routh reductions~\cite{MestdagCrampin}.

%%%%%%%%%%%%%

% BibTeX users please use one of
%\bibliographystyle{spbasic}      % basic style, author-year citations
%\bibliographystyle{spmpsci}      % mathematics and physical sciences
%\bibliographystyle{spphys}       % APS-like style for physics
%\bibliography{}   % name your BibTeX data base

% Non-BibTeX users please use

\end{document}